\documentclass[aps]{revtex4-2}

\usepackage{graphicx}
\usepackage{hyperref}
\usepackage{epstopdf}
\usepackage{setspace}
\usepackage{amsmath}
\usepackage{amsthm}
\usepackage{latexsym, amssymb}
\usepackage{tabularx, booktabs}
\usepackage{verbatim}
\usepackage{xcolor}
\usepackage{hhline}

\newcommand{\m}[1]{\mathcal{#1}}

\newcommand{\sgn}{\mathrm{sgn\,}}

\newcommand{\CR}{\nonumber\\}
\newcommand{\ket}[1]{|#1 \rangle}
\newcommand{\bra}[1]{\langle #1 |}
\newcommand{\braket}[2]{\langle #1 |#2 \rangle}

\theoremstyle{plain}

\newtheorem{lem}{Lemma}[section]

\theoremstyle{definition}
\newtheorem{defn}{Definition}[section]

\theoremstyle{remark}

\begin{document}

\begin{titlepage}
\title{Dynamical Quantum Multigraphs}
\author{Kassahun Betre}
\email{kassahun.betre@sjsu.edu}
\affiliation{San Jose State University, 1 Washington Square, San Jose, CA 95192, USA}
\author{Nathan Lewis}
\email{nathan.lewis@wvm.edu}
\affiliation{Mission College, 3000 Mission College Blvd, Santa Clara, CA 95054, USA}

\begin{abstract}
Motivated by applications in background-independent quantum gravity, we discuss the quantization of labeled and unlabeled finite multigraphs with a maximum edge count. We provide a unified way to represent quantum multigraphs with labeled or unlabeled vertices, which enables the study of quantum multigraphs as dynamical microscopic degrees of freedom and not just as representations of relations among quantum states of particles. The quantum multigraphs represent a quantum mechanical treatment of the relations themselves and give rise to Hilbert space realizations of relations. After defining the Hilbert space, we focus on quantum simple graphs and explore the thermodynamics resulting from two simple models, a free Hamiltonian and an Ising-type Hamiltonian (with interactions among nearest-neighbor edges). We show that removing the distinction among vertices by considering unlabeled vertices gives rise to a qualitatively different thermodynamics. We find that the free theory of labeled quantum simple graphs is the Erd\H{o}s--R\'enyi--Gilbert $G(N,p)$ model of random graphs. This model has analytic free energy and hence no thermodynamic phase transition. On the other hand, the unlabeled quantum graphs give rise to proper thermodynamic phase transitions in both the free and the ferromagnetic Ising models, characterized by divergence in the specific heat and critical slowing near the critical temperature. The thermodynamic phase transition has an order parameter given as the fraction of vertices in the largest connected component. Although this is similar to the phase transition in the $G(N,p)$ model, in this case it represents the actual thermodynamic phase transition. 
\end{abstract}
\maketitle
\vspace{-0.3in}
\begin{widetext}
\tableofcontents
\end{widetext}
\setcounter{page}{1}

\end{titlepage}

\section{Introduction}
Graphs are combinatorial objects with a set of vertices and edges that connect the vertices. Quantum versions of graphs, \emph{quantum graphs} or \emph{quantum networks}, have appeared in several distinct communities with different meanings. One type of quantum graphs was proposed and studied as one-dimensional approximations for waves propagating in thin structures \cite{kuchment2008quantum}. Such quantum graphs are made up of the Hilbert space of square integrable functions that live on the edge space of a weighted graph with boundary conditions on the vertices. In quantum information, quantum graphs represent entanglement states of multi-qubit systems. They arise as stabilizer states associated with a simple graph and underlie measurement‑based quantum computation and multipartite entanglement \cite{Hein2004GraphStates,Hein2006GraphStatesReview,plesch2003entangled}. A closely related generalization are \emph{quantum hypergraph states}, where multi‑qubit controlled‑phase gates encode hyperedges; they admit a generalized stabilizer description and a one‑to‑one correspondence with real equally weighted states used in algorithms \cite{Rossi2013NJP,Qu2013PRA,zhang2024quantum}. Quantum graphs also appear as ``noncommutative graphs" which generalize the confusability graphs of classical zero-error communication channels \cite{duan2012zero, musto2018compositional}. 

In quantum gravity, quantum graphs appear in at least two contexts. In loop quantum gravity, \emph{spin networks} and \emph{spin foams} furnish combinatorial state sums on graphs and 2‑complexes \cite{RovelliSmolin1995PRD,Baez1999SpinFoams}. In background-independent quantum gravity, specifically in Quantum Graphity, quantum graphs are defined to model the microscopic quantum states of a pre-geometric quantum gravity state without any reference to a background or auxiliary spacetime \cite{konopka2006quantum, konopka2008quantum}. In this context, the quantum graphs are not merely representations of entanglement states of particles, or operator spaces, but rather the underlying fundamental degrees of freedom on which a dynamics can be built. The graphs are dynamical as opposed to fixed. The Hilbert space of the system with $N$ vertices is the tensor product of $N(N-1)/2$ edge Hilbert spaces defined on each unordered pair of distinguishable (or labeled) vertices. The aim of Quantum Graphity is to build a model for locality, topology, and geometry of spacetime as emergent from dynamical quantum graphs. In contrast, in the PDEs and quantum information context, the quantum graphs, multigraphs, and hypergraphs do not represent relations that are themselves ``quantum" in the sense of existing in a Hilbert space. In addition, the underlying graphical or network structure is fixed and not dynamical.  

In this paper we take the notion of quantum graphs as defined in Quantum Graphity. We extend dynamical quantum graphs of Quantum Graphity to dynamical quantum multigraphs, and also include the treatment of indistinguishable (unlabeled) vertices. We then study the effect of removing the vertex labels on the thermodynamics of simple dynamical toy models inspired by the Erd\H{o}s--R\'enyi--Gilbert random graph model and the Ising model. Throughout this work, we treat the \emph{graph itself as a classical variable}, while each edge carries a finite‑dimensional local Hilbert space; our dynamics and thermodynamics are over these edge degrees of freedom and (for unlabeled ensembles) over isomorphism classes of graphs.

Broadly speaking ``quantization" of a classical system refers to a procedure of taking the states and observables of the classical system and mapping them to vectors in a Hilbert space and the self-adjoint operators acting on the Hilbert space. In that case, quantizing multigraphs, abstract simplicial complexes, or other combinatorial objects can be viewed similarly as finding Hilbert space realizations of these combinatorial objects. The main motivation to treating multigraphs quantum mechanically as the fundamental quantum degrees of freedom comes from possible applications in background-independent formulations of quantum gravity in which spacetime and its dynamics are expected to emerge as macroscopic average phenomena of combinatorially defined, non-geometric quantum microstates. Combinatorial objects such as graphs, abstract simplicial complexes, and matroids have been used for such purposes \cite{konopka2008quantum, lee2009emergence, hartle2022simplicial, brunnemann2010oriented, nieto2011oriented}. 

At a more foundational level, if one presupposes that everything in nature has a fundamentally quantum mechanical description, the list of entities with a quantum mechanical description should include not only ``stuff" such as matter and fields, but also relationships among constituents of ``stuff". Consider a system made up of $N$ constituents, for example $N$ particles. Ordinarily, the state of such a system is a vector in the $N$-fold tensor product Hilbert space of the one-particle Hilbert spaces of each particle. But now, there will also be a Hilbert space associated with every pair of particles that corresponds to the 2-way or 2-component relationships of pairs of particles. Going further, every 3-,4-,$\dots$, $k$-subset $\sigma_k$ of the $N$ constituents will have a corresponding Hilbert space describing the state of the 3- or 4- ... $k$-fold relationships. We emphasize here that the question of whether or not the higher-degree relations are derivable from the lower ones is a separate consideration. Let $\Delta$ be the clique complex over the complete graph $K_N$ with base set $[N]=\{1,2,\dots,N\}$ in which every non-empty subset of $[N]$ is a simplex, then, for each $k\in \{2,\dots,N\}$ the $k$-fold relational state is a vector in the Hilbert space $\mathcal{H}_{k}$. Therefore, the universal Hilbert space that describes the whole system with $N$ constituents together with all possible states of relationships will be 
\begin{align}
\mathcal{H}_{\mathcal{U}}=\bigotimes_{k=1}^N\mathcal{H}_{k},
\end{align}
where $\mathcal{H}_{1}$ is the familiar Hilbert space of the $N$-fold tensor product of the particles in many-body quantum mechanics. 

This signals an exponential expansion of the basic Hilbert space of a $N-$particle system. The dynamics governing the system will also expand exponentially. In addition to the traditional Hamiltonian operator governing the time evolution of the 1-particle space, there will be a Hamiltonian governing each of the 2-, 3-, \dots, $N-$way relational states as well as possible interaction Hamiltonians that cross between them. The same goes more generally to the algebra of operators acting on this Hilbert space. The aim of background-independent quantum gravity can then be embedded in this framework as seeking to produce semi-classical notions of spacetime geometry dynamically from quantized relations.

This paper focuses only on the Hilbert space of 2-way relations leading to quantum multigraphs. The outline of the paper is as follows. In Section \ref{sec:quantumGraphs}, after listing the classical definitions of multigraphs, we construct the edge Hilbert space of quantum multigraphs associated with 2-way relationships over labeled vertices. Section \ref{sec:unlabeledQuantumGraphs} defines the Hilbert space of unlabeled quantum multigraphs. In Section \ref{sec:dynamics}  we consider two simple dynamical toy models over the labeled and unlabeled quantum graphs. We show that removing vertex labels gives rise to a qualitatively different thermodynamics marked by thermodynamic phase transitions that are absent in the labeled case. Our result is consistent with a recent result by Evnin and Krioukov showing a first-order phase transition in unlabeled random networks with a fixed average number of links that does not exist in the labeled version \cite{EvninKrioukov2025PRL}. In Section \ref{sec:conclusion} we hypothesize the reason for this qualitative difference and for which types of systems the labeled and unlabeled thermodynamics converge. 

\section{Labeled Quantum Multigraphs}
\label{sec:quantumGraphs}
The goal of this section is to generally describe quantum versions of finite multigraphs with $N$ vertices where the maximum number of edges between any two vertices is fixed at $D$. The resulting space is a Hilbert space whose basis vectors can be put into a one-to-one correspondence with the set of multigraphs on $N$ nodes and maximum edge count of $D$ between any two vertices. The action of the symmetric group $S_N$ on classical multigraphs will translate analogously to a unitary action in the space of quantum multigraphs and will be utilized to define unlabeled multigraphs from orbit sums. We start the section with a short summary of classical multigraphs. The notation and terminology used in Section \ref{sec:definitions} follows the book ``Modern Graph Theory" \cite{bollobas1998modern}.

\subsection{Classical Multigraphs}
\label{sec:definitions}

A directed graph or {\em digraph} $G=(V,E)$ is a set of {\em vertices} $V$ together with a set $E\subseteq V\times V = \{(a,b)|a, b \in V\}$ of ordered pairs of vertices called {\em directed edges}. Edges of the form $(a,a)$ that start and end at the same vertex are called self-loops. A generalization of digraphs is a {\em directed multigraph} where multiple edges are allowed to start and end between any two vertices and multiple self-loops are also allowed. Formally, a {\em directed multigraph} $M=(V, E)$ is defined as a set of vertices $V$ together with a multiset $E$ (a ``set that allows elements to be repeated") of directed edges (ordered pairs of vertices). For {\em undirected multigraphs} the edge space is a multiset of unordered pairs of vertices. Self-loops and multiple edges are allowed in directed multigraphs. In this paper, we use the term ``multigraph" to refer to \emph{undirected multigraphs with no self-loops}. {\em Simple graphs} are undirected graphs that do not have self-loops or multiple edges. Directed and undirected multigraphs, digraphs, and simple graphs on $N$ vertices can be uniquely represented using a $N\times N$ matrix $A$ -- the {\em adjacency matrix} defined as an integer matrix whose $ij$th element is the number of edges going from the vertex $i$ to the vertex $j$.  The adjacency matrices of simple graphs and multigraphs are symmetric. Furthermore, for digraphs and simple graphs, the adjacency matrix elements are 1 or 0. Fig.(\ref{fig:classicalgraphs}) shows a directed multigraph, a digraph, and a simple graph together with their adjacency matrices. 

\begin{figure*}
\centering
\includegraphics[width=\textwidth]{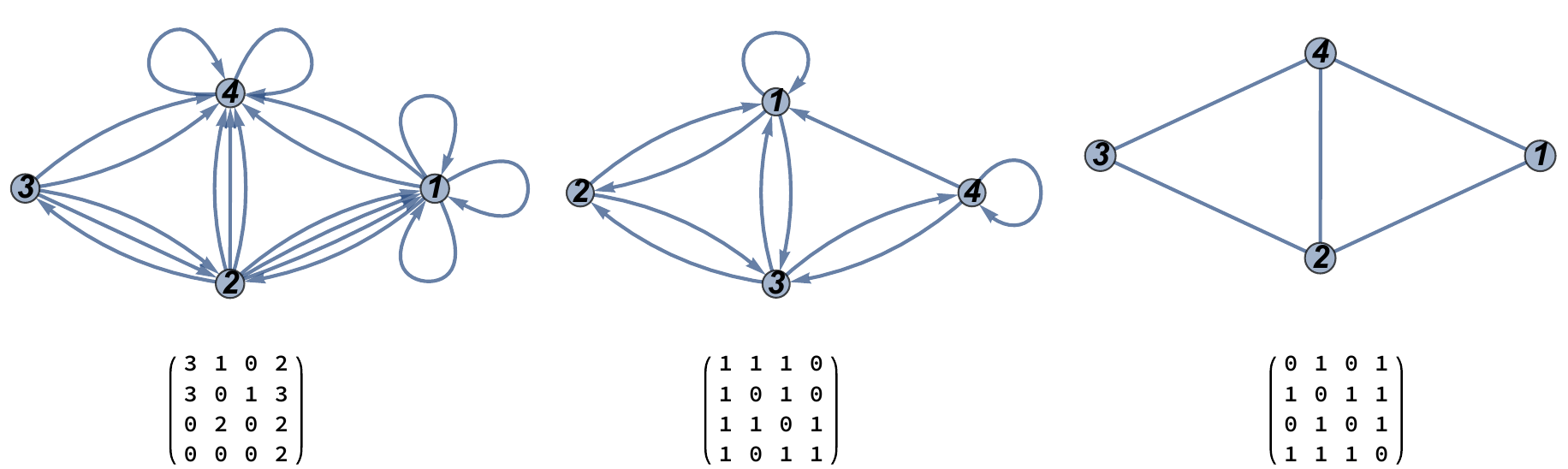}
\caption{A directed multigraph, digraph, and simple graph together with their adjacency matrices.}
\label{fig:classicalgraphs}
\end{figure*}

If we treat the vertices of a graph as distinguishable (or {\em labeled}), the number of all labeled digraphs with $N$ edges is $2^{N^2}$. There are $N^2$ adjacency matrix elements $A_{ij}$ each having the possibility of being 1 or 0. For simple graphs the count is $2^{N(N-1)/2}$. The number of directed multigraphs on $N$ nodes and a maximum edge count $D$ between any two vertices is $D^{N^2}$. The number of undirected multigraphs with no self-loops is $D^{N(N-1)/2}$.

Two multigraphs $G_1(V_1,E_1), G_2(V_2,E_2)$ with $N$ vertices are said to be {\em isomorphic} if there is a set map $f:V_1\rightarrow V_2$ of the vertex sets such that for any edge $\{i,j\}\in E_1$ the image $\{f(i),f(j)\}\in E_2$ is an edge of $G_2$. 
An {\em automorphism} of a multigraph $G(V,E)$ is an isomorphism from the multigraph back to itself, i.e., a permutation $P$ of the vertex set $V$ is an automorphism of $G$ if for all edges $\{i,j\}\in E$, $\{P(i),P(j)\}\in E$. The set of automorphisms of a multigraph forms a group under composition: the {\em automorphism group} of $G$, denoted by $\Gamma(G)$. 
\begin{align}
\label{eq:automorphismGroup}
	\Gamma(G) = \{P \in S_N| P(G)=G\}.
\end{align}
Let $G(V,E)$ be an unlabeled multigraph. The number of ways of labeling $G$ to produce all unique labeled multigraphs is $|V(G)|!/|\Gamma(G)|$, where $|\Gamma(G)|$ is the order of the automorphism group of $G$. So, for example, the number of all labeled graphs on the $N$ vertices is
\begin{align}
		2^{N(N-1)/2} = N!\sum_{G_u}\frac{1}{|\Gamma(G_u)|},
\end{align}
where $\{G_u\}$ is the set of unlabeled graphs with $N$ vertices \cite{harary2014graphical}.

\subsection{Labeled Quantum Directed Multigraphs}
We now define quantum versions of directed multigraphs that live in a Hilbert space. The construction is an extension of quantum simple graphs defined in \cite{konopka2006quantum,konopka2008quantum}. Consider $N$ labeled vertices $V=[N]\equiv\{1,2,\dots,N\}$ and introduce a ``single-particle" edge Hilbert space $\m{H}_{ij}$ associated with each ordered pair $(i,j)\in V\times V$ of vertices. We are concerned only with the case where all edge Hilbert spaces $\m{H}_{ij}$ are copies of the same single-particle Hilbert space $\m{H}$. Let the single particle Hilbert space have complex dimensions $\dim \m{H} = D$. The total Hilbert space of the quantum directed multigraph (with self-loops), $\m{H}_{tot}$ is the tensor product over each edge Hilbert space. 
\begin{align}
    \m{H}_{tot} = \bigotimes_{i,j=1}^N\m{H}_{ij}
\end{align}
If the one-particle Hilbert space is spanned by a set of basis vectors $\m{H}=\mathrm{Span}\{\ket{n}\},n\in\{0,1,\dots,D-1\}$, then the total Hilbert space is spanned by basis vectors 
\begin{align}
    \m{H}_{tot}=\mathrm{Span}\left\{\prod_{i,j=1}^N\ket{n_{ij}}_{ij}\right\},
\end{align}
where each $n_{ij}$ independently takes values in $\{0,1,\dots,D-1\}$. The dimension of the total Hilbert space is $\dim \m{H}_{tot} = D^{N^2}$. Each of the $D^{N^2}$ tensor product basis vectors can be put in one-to-one correspondence with directed multigraphs with self-loops and maximum edge count $D$. We declare this Hilbert space to be the quantum version of a directed multigraph with self-loops on $N$ vertices. Fixing $D=2$ gives rise to quantum digraphs with self-loops.  
Without loss of generality, we will always write the tensor product basis states in lexicographic order. Examples of a quantum digraph and a quantum directed multigraph with $N=3$ are shown in Fig.(\ref{fig:quantumGraphs}) for the case where the single-particle edge states are $\m{H}_{ij}=\mathrm{Span}\{\ket{0},\ket{1}\}$ and $\m{H}_{ij}=\mathrm{Span}\{\ket{0},\ket{1},\ket{2}\}$.  

\begin{figure*}
\centering

 \includegraphics[width=0.7\textwidth]{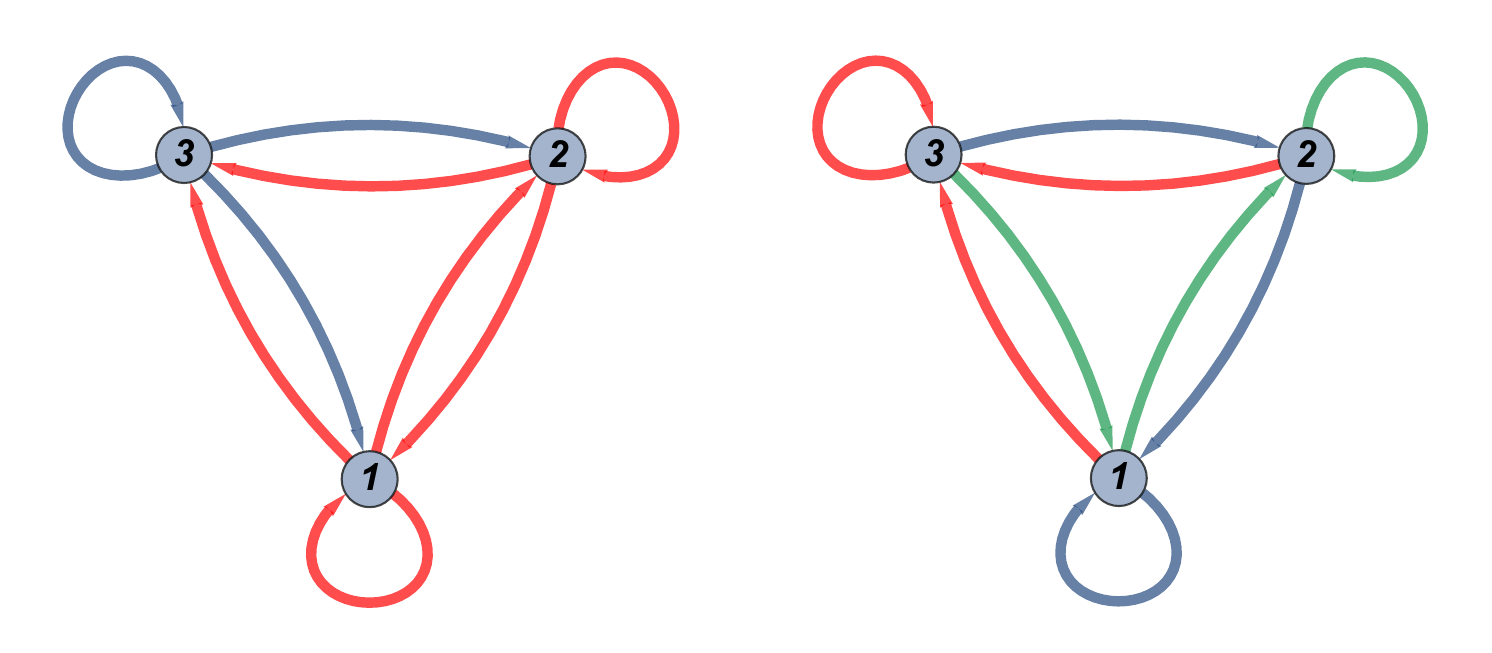}
 
\caption{(Left) A diagrammatic representation of a quantum digraph with $N=3$ and $\m{H}_{ij}=\mathrm{Span}\{\ket{0},\ket{1}\}$ given by $\ket{G}=\ket{0}_{11}\ket{0}_{12}\ket{0}_{13}\ket{0}_{21}\ket{0}_{22}\ket{0}_{23}\ket{1}_{31}\ket{1}_{32}\ket{1}_{33}$. Red edges represent $\ket{0}$ and blue $\ket{1}$. (Right) A diagrammatic representation of a quantum directed multigraph for $N=3$ and $\m{H}_{ij}=\mathrm{Span}\{\ket{0},\ket{1},\ket{2}\}$ given by the state $\ket{G}=\ket{1}_{11}\ket{2}_{12}\ket{0}_{13}\ket{1}_{21}\ket{2}_{22}\ket{0}_{23}\ket{2}_{31}\ket{1}_{32}\ket{0}_{33}$. Red edges represent $\ket{0}$, blue $\ket{1}$, and green $\ket{2}$.}
\label{fig:quantumGraphs}
\end{figure*}

There is a natural action of the permutation group $S_N$ on this Hilbert space defined as permutation of the edge labels followed by rearranging the tensor products back to lexicographic order. Let $\pi\in S_N$ be a permutation; then
\begin{widetext}
\begin{align}
     \pi\left(\ket{n_{11}}_{11}\dots\ket{n_{NN}}_{NN}\right) &= \ket{n_{11}}_{\pi(1)\pi(1)}\dots\ket{n_{11}}_{\pi(N)\pi(N)}\CR
     &=\ket{n_{\pi^{-1}(1)\pi^{-1}(1)}}_{11}\ket{n_{\pi^{-1}(1)\pi^{-1}(2)}}_{12}\dots\ket{n_{\pi^{-1}(N)\pi^{-1}(N)}}_{NN}
\end{align}
\end{widetext}

For example, the permutations $\pi_1=(12)$, and $\pi_2=(23)$ acting on the state \[\ket{G}=\ket{0}_{11}\ket{0}_{12}\ket{0}_{13}\ket{0}_{21}\ket{0}_{22}\ket{0}_{23}\ket{1}_{31}\ket{1}_{32}\ket{1}_{33}\] in Fig.(\ref{fig:quantumGraphs}, Left) give,
\begin{widetext}
\begin{align}
    (12)\ket{G}&=\ket{0}_{22}\ket{0}_{21}\ket{0}_{23}\ket{0}_{12}\ket{0}_{11}\ket{0}_{13}\ket{1}_{32}\ket{1}_{31}\ket{1}_{33}\CR
    &=\ket{0}_{11}\ket{0}_{12}\ket{0}_{13}\ket{0}_{21}\ket{0}_{22}\ket{0}_{23}\ket{1}_{31}\ket{1}_{32}\ket{1}_{33}\CR
     &=\ket{G}\CR    (23)\ket{G}&=\ket{0}_{11}\ket{0}_{13}\ket{0}_{12}\ket{0}_{31}\ket{0}_{33}\ket{0}_{32}\ket{1}_{21}\ket{1}_{23}\ket{1}_{22}\CR    &=\ket{0}_{11}\ket{0}_{12}\ket{0}_{13}\ket{1}_{21}\ket{1}_{22}\ket{1}_{23}\ket{0}_{31}\ket{0}_{32}\ket{0}_{33}\CR
     &=\ket{G'}\neq\ket{G}
\end{align}
\end{widetext}
The action of the first permutation $\pi_1=(12)$ on $\ket{G}$ resulted in $\ket{G}$ itself, but $\pi_2=(23)$ produced a different state $\ket{G'}$. So $\pi_1$ is a symmetry of state $\ket{G}$. 
We take the automorphism group of a quantum directed multigraph to be the set of all permutations $\pi\in S_N$ such that $\pi\ket{G}=\ket{G}$. The set of automorphisms $\pi\in S_N$ of a quantum multigraph $\ket{G}$ forms a group under composition, the automorphism group $\Gamma(\ket{G})=\{\pi\in S_N : \pi\ket{G}=\ket{G}\}$.

This action of $S_N$ on $\m{H}_{tot}$ gives rise to a $D^{N^2}\times D^{N^2}$ dimensional unitary (permutation) representation of the symmetric group $S_N$, and the automorphism group $\Gamma(\ket{G})$ is similarly $D^{N^2}\times D^{N^2}$ dimensional representation of the automorphism group of the classical directed multigraph.

\subsection{Quantum Undirected Multigraphs with No Self-loops}

Now we consider the case of undirected multigraphs with no self-loops which, as stated above, we will simply refer to as multigraphs. To define their Hilbert space, we will need to remove states with self-loops and restrict down to only one Hilbert space $\mathcal{H}_{ij}$ per unordered pair $\{i,j\}$ of vertices. We again start with $N$ labeled vertices $V=[N]\equiv\{1,2,\dots,N\}$ and introduce a single-particle Hilbert space $\m{H}_{ij}$ associated with each unordered pair $\{i,j\}\subseteq V$ of vertices. As before, we are concerned only with the case where all edge Hilbert spaces $\m{H}_{ij}$ are copies of the same single-particle Hilbert space $\m{H}$. The total quantum multigraph Hilbert space is the tensor product,
\begin{align}
    \m{H}_{MG}=\bigotimes_{j>i=1}^N\m{H}_{ij}
\end{align}
If the single-particle Hilbert space has dimension $D$, $\m{H}_{MG}$ is $D^{N(N-1)/2}$ dimensional.  Given a basis $\{\ket{n}\}$ for the single-particle Hilbert space $\m{H}$, the total Hilbert space is spanned by the tensor product kets, 
\begin{align}
    \m{H}_{MG}=\mathrm{Span}\{\prod_{j>i=1}^N\ket{n_{ij}}_{ij}\}
\end{align}
As before, we order the kets in lexicographic order of the labels. Each of the basis states can be put in a one-to-one correspondence with the set of all labeled multigraphs on $N$ vertices. Therefore, the Hilbert space $\m{H}_{MG}$ provides the quantization of finite labeled undirected multigraphs with no self-loops and a maximum edge of $D$.
\begin{figure*}
\centering
 \includegraphics[width=0.7\textwidth]{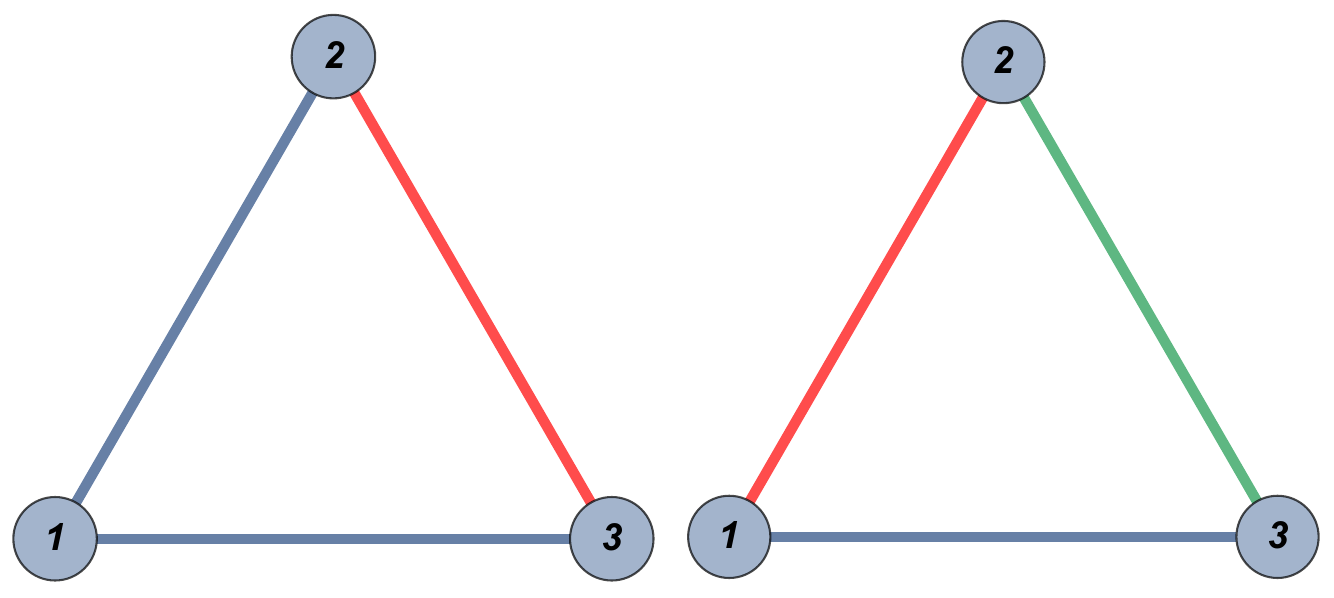}
\caption{(Left) A diagrammatic representation of a quantum multigraph with $N=3$ and $\m{H}_{ij}=\mathrm{Span}\{\ket{0},\ket{1}\}$ given by $\ket{G}=\ket{1}_{12}\ket{1}_{13}\ket{0}_{23}$. Red edges represent $\ket{0}$ and blue $\ket{1}$. (Right) A diagrammatic representation of a quantum multigraph for $N=3$ and $\m{H}_{ij}=\mathrm{Span}\{\ket{0},\ket{1},\ket{2}\}$ given by the state $\ket{G}=\ket{0}_{12}\ket{1}_{13}\ket{2}_{23}$. Red edges represent $\ket{0}$, blue $\ket{1}$, and green $\ket{2}$.}
\label{fig:quantumUndirectedGraphs}
\end{figure*}
The symmetric group $S_N$ acts on $\m{H}_{MG}$ the same way it did on the quantum directed multigraph Hilbert space $\m{H}_{tot}$ with one modification. Since the undirected quantum graph single-particle edge kets $\ket{n}_{ij}$ are always indexed with $j>i$, we have to determine what to do when we encounter kets such as $\ket{n}_{ij}$ with $i>j$ at the end of acting with a permutation. There are two natural choices that produce an action of $S_N$, $\ket{n}_{ij}=\pm\ket{n}_{ji}$, so that, for example, $\ket{0}_{21}=\pm\ket{0}_{12}$. We distinguish between the two cases by putting square brackets or parentheses around the subscripts:  $\ket{n_{ij}}_{(ij)}=\ket{n_{ij}}_{(ji)}$, and $\ket{n_{ij}}_{[ij]}=-\ket{n_{ij}}_{[ji]}$. We will refer to these kets as even states ($\ket{n_{ij}}_{(ij)}$) and odd states($\ket{n_{ij}}_{[ij]}$). The total undirected quantum graph Hilbert space corresponding to each choice of sign will be referred to as symmetric $\m{H}^S_{MG}=\mathrm{Span}\{\prod_{j>i=1}^N\ket{n_{ij}}_{(ij)}\}$ or antisymmetric $\m{H}^A_{MG}=\mathrm{Span}\{\prod_{j>i=1}^N\ket{n_{ij}}_{[ij]}\}$. As an example, consider a state $\ket{G}^A=\ket{1}_{[12]}\ket{1}_{[13]}\ket{0}_{[23]}$ in the antisymmetric Hilbert space $\m{H}^A_{MG}$ represented diagrammatically on the right hand side of Fig.(\ref{fig:quantumUndirectedGraphs}), and two permutations $\pi_1=(12), \pi_2=(23)$ in $S_3$, then
\begin{align}
    (12)\ket{G}^A&=\ket{1}_{[21]}\ket{1}_{[23]}\ket{0}_{[13]}\CR
    &=-\ket{1}_{[12]}\ket{0}_{[13]}\ket{1}_{[23]}\CR     (23)\ket{G}^A&=\ket{1}_{[13]}\ket{1}_{[12]}\ket{0}_{[32]}\CR  &=-\ket{1}_{[12]}\ket{1}_{[13]}\ket{0}_{[23]}=-\ket{G}
\end{align}
The action of $(23)$ resulted in the same ket up to a phase, and so is a symmetry of the state. We define automorphisms and the automorphism group of the antisymmetric quantum multigraphs as $\Gamma(\ket{G}^A)=\{\pi\in S_N: \pi\ket{G}^A=\mathrm{sgn}(\pi)\ket{G}^A\}$. In the case of the symmetric states, $\Gamma(\ket{G}^S)=\{\pi\in S_N: \pi\ket{G}^S=\ket{G}^S\}$.

\subsection{Ladder Operators}
Let the one-particle $D$-dimensional Hilbert space of the edge states be spanned by orthonormal basis vectors $\m{H}=\mathrm{Span}\{\ket{n}\},n\in\{0,1,\dots,D-1\}$ where an arbitrary ordering of the basis kets is implied, analogous to the treatment of angular momentum eigenstates in quantum mechanics we define a ``raising" ladder operator $L^+$ and a ``lowering" ladder operator $L^-$ as the operators with the action, 
\begin{align}
L^+\ket{n}
=
\begin{cases}\ket{n+1},& n<D-1\\
0,&n=D-1
\end{cases},\quad L^-\ket{n}=
\begin{cases}
\ket{n-1},&n>0\\
0,&n=1
\end{cases}
\end{align}

\noindent In the basis $\{\ket{0},\ket{1},\dots,\ket{D-1}\}$, these operators have matrix representations. 
\begin{align}
    \left(L^+\right)_{nm}&=\begin{cases}
    1,&\text{ if }n=m-1\\
    0&\text{otherwise}
\end{cases},\qquad
\left(L^-\right)_{nm}=\begin{cases}
    1,&\text{ if }n=m+1\\
    0&\text{otherwise}
\end{cases}
\end{align}

\noindent They have the following commutation and anticommutation relations:
\begin{align}
    \{L^+,L^-\}_{nm}&=2\delta_{nm}-\delta_{n0}-\delta_{m,D-1}\\
    [L^+,L^-]_{nm}&=\delta_{n0}-\delta_{m,D-1}
\end{align}

\noindent We promote these operators to act on the quantum directed multigraph Hilbert space $\m{H}_{tot}$. There are $N^2$ pairs of raising and lowering operators $L^+_{ij}, L^-_{ij}$ associated with each ordered pair $(i,j)$ of vertices given by
\begin{align}
    L^+_{ij}\equiv I_{11}\otimes\dots\otimes L^+_{ij}\otimes\dots\otimes I_{NN}
\end{align}
By abuse of notation, we use $L^+_{ij}$ to refer to both the operator in the quantum directed multigraph Hilbert space $\m{H}_{tot}$ and the one-particle Hilbert space of the $ij$th edge.
The raising and lowering operators commute if their indices are not identical, so, in general, these operators obey the following sets of commutation and anticommutation relations. 
\begin{align}
    [L^+_{ij},L^+_{kl}]&=[L^-_{ij},L^-_{kl}]=0\CR
    {[L^+_{ij},L^-_{kl}]}_{mn}&=\delta_{ik}\delta_{jl}\left(\delta_{n0}-\delta_{m,D-1}\right)\CR
     \{L^+_{ij},L^-_{kl}\}_{mn}&=\delta_{ik}\delta_{jl}\left(2\delta_{nm}-\delta_{n0}-\delta_{m,D-1}\right)
\end{align}

\noindent The raising and lowering operators act on the quantum directed multigraph basis kets as follows:
\begin{align}
L^+_{ij}&\Big(\ket{n_{11}}_{11}\dots\ket{n_{ij}}_{ij}\dots\ket{n_{NN}}_{NN}\Big)\CR
&=\begin{cases}
    0&\text{ if }n_{ij}=D-1\\
    \ket{n_{11}}_{11}\dots\ket{n_{ij}+1}_{ij}\dots\ket{n_{NN}}_{NN}&\text{ otherwise }
\end{cases}
\CR
L^-_{ij}&\Big(\ket{n_{11}}_{11}\dots\ket{n_{ij}}_{ij}\dots\ket{n_{N}}_{NN}\Big)
\CR
&=\begin{cases}
    0&\text{ if }n_{ij}=0\\
    \ket{n_{11}}_{11}\dots\ket{n_{ij}-1}_{ij}\dots\ket{n_{N}}_{NN}&\text{ otherwise }
\end{cases}
\end{align}
The construction is carried out analogously to the symmetric and antisymmetric quantum multigraph (i.e. undirected with no self-loops) Hilbert spaces. There are $N(N-1)/2$ pairs of raising and lowering operators ${L^S_{ij}}^+, {L^S_{ij}}^-$, and ${L^A_{ij}}^+, {L^A_{ij}}^-$ so that
\begin{align}
{L^S_{ij}}^+&\Big(\ket{n_{12}}_{(12)}\dots\ket{n_{ij}}_{(ij)}\dots\ket{n_{N-1,N}}_{(N-1,N)}\Big)\CR
&=\begin{cases}
    0&\text{ if }n_{ij}=D-1\\
    \ket{n_{12}}_{(12)}\dots\ket{n_{ij}+1}_{(ij)}\dots\ket{n_{N-1,N}}_{(N-1,N)}&\text{ otherwise }
\end{cases}\CR
{L^S}^-_{ij}&\Big(\ket{n_{12}}_{(12)}\dots\ket{n_{ij}}_{(ij)}\dots\ket{n_{N-1,N}}_{(N-1,N)}\Big)\CR
&=\begin{cases}
    0&\text{ if }n_{ij}=0\\
    \ket{n_{12}}_{(12)}\dots\ket{n_{ij}-1}_{(ij)}\dots\ket{n_{N-1,N}}_{(N-1,N)}&\text{ otherwise }
\end{cases}
\end{align}
Similar relations hold for the antisymmetric raising and lowering operators ${L^+}^A_{ij}, {L^-}^A_{ij}$ and the antisymmetric basis kets $\{\ket{n_{ij}}_{[ij]}, i<j\in[N]\}$. These also commute for non-identical indices.  

Define the ``ground states" $\ket{0}, \ket{0}^A,\ket{0}^S$ of $\m{H}_{tot}, \m{H}^A_{MG}$ and $\m{H}^S_{MG}$ respectively to be
\begin{align}
    \ket{0}&=\ket{0}_{11}\ket{0}_{12}\dots\ket{0}_{NN}\CR
    \ket{0}^S&=\ket{0}_{(12)}\ket{0}_{(13)}\dots\ket{0}_{(N-1,N)}\CR
    \ket{0}^A&=\ket{0}_{[12]}\ket{0}_{[13]}\dots\ket{0}_{[N-1,N]}
\end{align}
Note that under the action of the symmetric group, the directed ground state $\ket{0}$ and the symmetric undirected ground state $\ket{0}^S$  are fully symmetric, and the antisymmetric undirected ground state $\ket{0}^A$ is fully antisymmetric since exactly the same number of kets as the number of transpositions of the permutation will pick up minus signs:
\begin{align}
    \pi\ket{0}&=\ket{0},\;\;\pi\ket{0}^S=\ket{0}^S,\;\; \pi\ket{0}^A=\sgn(\pi)\ket{0}^A
\end{align}
We can use the raising and lowering operators to express any basis ket as the action of the raising operators acting on the ground state.

\begin{align}
    \ket{G}&=\ket{n_{11}}_{11}\dots\ket{n_{NN}}_{NN}=\left(L^+_{11}\right)^{n_{11}}\dots\left(L^+_{NN}\right)^{n_{NN}}\ket{0}\CR
    \ket{G}^S&=\ket{n_{11}}_{(11)}\dots\ket{n_{NN}}_{(NN)}=\left({L^S}^+_{11}\right)^{n_{11}}\dots\left({L^S_{NN}}^+\right)^{n_{NN}}\ket{0}^S\CR
    \ket{G}^A&=\ket{n_{12}}_{[12]}\dots\ket{n_{N-1,N}}_{[N-1,N]}=\left({L^A}^+_{12}\right)^{n_{12}}\dots\left({L^A}^+_{N-1,N}\right)^{n_{N-1,N}}\ket{0}^A
\end{align}

The action of the permutation group $S_N$ on the quantum multigraph basis kets implies the following permutation properties of the raising and lowering operators. Let $\pi\in S_N$ be a permutation, and we will use the same symbol for its unitary representation that acts on the Hilbert spaces. First, for the operators of the directed multigraph Hilbert space,
\begin{align}
    \pi \left(L^+_{ij}\ket{0}\right)&=\left(\pi L^+_{ij}\pi^{-1}\right)\pi\ket{0}\CR
    &=\left(\pi L^+_{ij}\pi^{-1}\right)\ket{0}\CR
    &=L^+_{\pi(i)\pi(j)}\ket{0}
\end{align}
Similarly,
\begin{align}
    \pi \left({L^S_{ij}}^+\ket{0}^S\right)&=\left(\pi {L^S_{ij}}^+\pi^{-1}\right)\pi\ket{0}^S\CR
    &=\left(\pi {L^S_{ij}}^+\pi^{-1}\right)\ket{0}^S,\quad\text{but}\CR
    \pi \left({L^S_{ij}}^+\ket{0}^S\right)&= \pi\ket{0}_{(12)}\dots\ket{1}_{(ij)}\dots\ket{0}_{(N-1,N)} \CR
    &=\ket{0}_{(\pi(1)\pi(1))}\dots\ket{1}_{(\pi(i)\pi(j))}\dots\ket{0}_{(NN)}\CR
    &={L^S}^+_{\big(\pi(i)\pi(j)\big)}\ket{0}^S,
\end{align}
where 
\begin{align}
    {L^S}^+_{(kl)}=\begin{cases}
        {L^S}^+_{kl}&k\le l\\
        {L^S}^+_{lk}&k>l
    \end{cases}
\end{align}
Therefore, 
\begin{align}
    \pi {L^S}^+_{ij}\pi^{-1}={L^S}^+_{\big(\pi(i)\pi(j)\big)}.
\end{align}
In particular, let $i<j$ and take $\pi = (ij)$, we find $ \pi {L^+}^S_{ij}\pi^{-1}={L^+}^S_{ji} = {L^+}^S_{ij}$. So, the raising (and lowering) operators of $\m{H}_{tot}$ and $\m{H}^S_{MG}$ are symmetric under the permutation of their labels. 

For the antisymmetric undirected Hilbert space,
\begin{align}
    \pi \left({L^A}^+_{ij}\ket{0}^A\right)&=\pi\ket{0}_{[12]}\dots\ket{1}_{[ij]}\dots\ket{0}_{[N,N-1]}\CR
    &=\ket{0}_{[\pi(1)\pi(2)]}\dots\ket{1}_{[\pi(i)\pi(j)]}\dots\ket{0}_{[\pi(N-1)\pi(N)]}\CR
    &=\sgn(\pi)\ket{0}_{[12]}\dots\ket{1}_{[\pi_<(ij)]}\dots\ket{0}_{[N,N-1]}\CR
    &=\sgn(\pi){L^A}^+_{\left(\pi(i)\pi(j)\right)}\ket{0}^A,
\end{align}    
where in the third line we used the notation $[\pi_<(ij)]$ to mean 
\begin{align}
    [\pi_<(ij)] =\begin{cases}
         [\pi(i)\pi(j)],&\text{ if }\pi(i) < \pi(j)\\
         [\pi(j)\pi(i)],&\text{ if }\pi(i) > \pi(j)
    \end{cases}
\end{align}
For example, for $N=3$, $\pi = (13)$ and $(ij)=(12)$, $[\pi_<(12)]=[23]$, i.e., 
\begin{align}
    (13)\left({L^A}^+_{12}\ket{0}^A\right)&=(13)\ket{1}_{[12]}\ket{0}_{[13]}\ket{0}_{[23]}\CR
    &=\ket{1}_{[32]}\ket{0}_{[31]}\ket{0}_{[21]}\CR
    &=\sgn((13))\ket{0}_{[12]}\ket{0}_{[13]}\ket{1}_{[23]}\CR
    &=\sgn((13)){L^A}^+_{23}\ket{0}^A
\end{align}
On the other hand, 
\begin{align}
\pi \left({L^A}^+_{ij}\ket{0}^A\right) &=\left(\pi {L^A}^+_{ij}\pi^{-1}\right)\pi\ket{0}^A\CR
    &=\left(\pi {L^A}^+_{ij}\pi^{-1}\right)\sgn(\pi)\ket{0}^A
\end{align}
Therefore, 
\begin{align}
    \pi {L^A}^+_{ij}\pi^{-1} = {L^A}^+_{\left(\pi(i)\pi(j)\right)}
\end{align}
In particular, let $i<j$ and take $\pi = (ij)$, we find $ \pi {L^A}^+_{ij}\pi^{-1}={L^A}^+_{ji} = {L^A}^+_{ij}$. So, the antisymmetric raising and lowering operators are also symmetric under the permutation of their indices. In summary, the ladder operators for all cases are symmetric under the permutation of their indices, hence there is no need to distinguish their indices with parentheses or square brackets.

\subsection{Occupation Graph Basis}
In many-body quantum mechanics, the occupation number basis is constructed as a suitable basis for the treatment of identical particles. We will find it convenient to also use a similar basis. At this point, however, we are still working with quantum multigraphs with labeled vertices, and the analogous occupation basis, though available, is not commonly used for distinguishable particles in many-body quantum mechanics. 

The new basis, which we call the {\em occupation graph} basis, is constructed as follows. We take each of the tensor product basis kets of the quantum multigraphs (ordered lexicographically) and order them differently as follows; we first list all edge states that are in the single-particle state of $\ket{0}$, then we list all edge states that are in the single-particle state $\ket{1}$, and so forth until finally we list all edge states in the single-particle state $\ket{D-1}$. For example, for $N=3$ case of the directed multigraph (with self-loops) Hilbert space $\m{H}_{tot}$, take the tensor product ket  $\ket{G}=\ket{1}_{11}\ket{2}_{12}\ket{0}_{13}\ket{1}_{21}\ket{2}_{22}\ket{0}_{23}\ket{2}_{31}\ket{1}_{32}\ket{0}_{33}$ shown in Fig.\ref{fig:quantumGraphs}. This will be mapped to $\ket{0}_{13}\ket{0}_{23}\ket{0}_{33}\ket{1}_{11}\ket{1}_{21}\ket{1}_{32}\ket{2}_{12}\ket{2}_{22}\ket{2}_{31}$. 
We represent this new ket as \[\ket{\{13,23,33\},\{11,21,32\},\{12,22,31\}}=\ket{G_0,G_1,G_2},\] where $G_i$ is a graph made up of all edges in the single-particle state $\ket{i}, i\in\{0,1,\dots,D_1\}$. Clearly, this is a bijection between the basis kets and, therefore, is just a different basis of the quantum multigraph Hilbert space. In general, if the single particle Hilbert space is $D$ dimensional, the quantum multigraph Hilbert space with $N$ vertices is spanned by kets of the form $\ket{G_0,G_1,\dots,G_{d-1}}$, where $(G_0,G_1,\dots,G_{d-1})$ is an ordered weak set partition of $[N]\times[N]$ (meaning that the empty set is allowed in the partition). We can count the occupation graph basis kets using Stirling numbers of the second kind $S(n,m)$, which count the number of unordered partitions of a set of size $n$ into exactly $m$ parts:

\begin{align}
    \sum_{k=0}^{D-1}\binom{D}{k}(D-k)!S(N^2,D-k)
    &=\sum_{k=0}^{D}\left(D\right)_{D-k}S(N^2,D-k)\CR
    &=\sum_{k'=0}^{D}\left(D\right)_{k'}S(N^2,k')=D^{N^2}
\end{align}
In the first line, we have grouped the weak partitions by the number of empty parts ($k$). There are $\binom{D}{k}$ ways to choose $k$ empty parts and $(D-k)!S(N^2,D-k)$ ways to make $D-k$ ordered partitions of a set with $N^2$ elements. In the last line, we used the property of Stirling numbers of the second kind as the change of basis matrices, $\sum_k (x)_k S(n,k) = x^n$, from the falling factorials $(x)_n = x(x-1)\dots (x-n+1)$ to the powers $x^n$ and the fact that $S(n,m)=0\;\; \forall m>n, \;\;(n)_m=0\;\;\forall m>n$. 

The same kind of occupation graph basis can be used for the symmetric and antisymmetric quantum multigraph Hilbert spaces $\m{H}^S_{MG},\m{H}^A_{MG}$.
\begin{align}
    \m{H}^{S/A}_{MG}&=\mathrm{Span}\{\ket{G_0,G_1,\dots G_{D-1}}^{S/A}\}
\end{align}
where now $(G_0,G_1,\dots G_{D-1})$ is a weak ordered partition of $\binom{[N]}{2}$, two-combinations of the set $[N]=\{1,2,\dots,N\}$, with the edges $\{i,j\}$ in each $G_k$ listed in increasing order $i<j$. Taking $G_k$ as graphs with vertex set $[N]$ and edge set $G_k$, we find that $(G_0,G_1,\dots G_{D-1})$ is an ordered weak partition of the edge set of the complete graph $K_N$.

We list the following observations.
\begin{itemize}
    \item The action of the symmetric group on the occupation graph bases is 
    \begin{align}
        \pi\ket{G_0,\dots,G_{D-1}}&=\ket{\pi(G_0),\dots,\pi(G_{D-1})}\CR
        \pi\ket{G_0,\dots,G_{D-1}}^S&=\ket{\pi(G_0),\dots,\pi(G_{D-1})}^S,\CR
        \pi\ket{G_0,\dots,G_{D-1}}^A&=\sgn(\pi)\ket{\pi(G_0),\dots,\pi(G_{D-1})}^A
    \end{align} where $\pi(G_i)$ is the usual action of the  permutation on a classical multigraph. Note that permutations do not change the total number of edges in each single particle state. As a result, the symmetric group does not act transitively on the quantum multigraph Hilbert spaces.     
    \item A permutation $\pi$ is an automorphism of a state in occupation basis if 
    \begin{align}
        \pi\ket{G_0,\dots,G_{D-1}}&=\ket{G_0,\dots,G_{D-1}},\CR
        \pi\ket{G_0,\dots,G_{D-1}}^S&=\ket{G_0,\dots,G_{D-1}}^S,\CR
        \pi\ket{G_0,\dots,G_{D-1}}^A&=\sgn(\pi)\ket{G_0,\dots,G_{D-1}}^A
    \end{align}
    So the automorphism group of a quantum graph $\ket{G}=\ket{G_0,\dots,G_{D-1}}$ is the intersection of automorphism groups of $G_0,\dots,G_{D-1}$, i.e., \[\Gamma(\ket{G})=\Gamma(G_0)\cap\dots\cap\Gamma(G_{D-1})\]
    \item The raising operator $L^+_{ij}$ acts on the occupation graph basis kets by shifting the edge $\{i,j\}$ one level up the ladder \[G_0\rightarrow G_1\rightarrow\dots\rightarrow G_{D-1}\rightarrow 0.\] Similarly the lowering operator $L^-_{ij}$ moves edge $\{i,j\}$ down the ladder \[0\leftarrow G_0\leftarrow G_1\leftarrow\dots\leftarrow G_{D-1}\] For example, 
    \begin{align}
        \left(L^+_{12}\right)^3\ket{\{\{12,13,14\},\{\},\{\}\}\}}&=\left(L^+_{12}\right)^2\ket{\{\{13,14\},\{12\},\{\}\}\}}\CR
        &=\left(L^+_{12}\right)\ket{\{\{13,14\},\{\},\{12\}\}\}}\CR
        &=0
    \end{align}
\end{itemize}
\emph{Remark}: The occupation graph basis can be generalized in a straightforward manner to discuss Hilbert spaces of $k$-fold relations. The basis kets of these Hilbert spaces will be weak ordered partitions of the $k$-combinations of $[N]$ into $D$ blocks. These basis kets are in one-to-one correspondence with uniform hypergraphs or $k$-pure simplicial complexes with $N$ vertices.

\subsection{Indicator and Number Operators}
Here we define some operators that will prove useful in later discussions.

 

\begin{defn}
\label{defn.indicator}
    The one particle {\bf indicator} operator $\m{I}^k_{ij}$ defined as 
    \begin{align}
        \m{I}^k_{ij}=\left(L^+_{ij}\right)^k\left(L^-_{ij}\right)^{D-1}\left(L^+_{ij}\right)^{D-1-k},\;\; k\in \{0,1,\dots,D-1\}
        \label{eq:oneParticleNumber}
    \end{align} acting on any occupation graph basis ket gives 1 or 0 depending on whether or not edge $ij$ is in the $k$th graph $G_k$, i.e., 
    \begin{align}
        \m{I}^k_{ij}\ket{G_0,\dots,G_{D-1}}=\begin{cases}
            \ket{G_0,\dots,G_{D-1}}&\text{ if } ij \in G_k\\
            0 & \text{ otherwise }
        \end{cases}
    \end{align}
    Therefore, the one-particle {\bf edge occupation number} $\m{N}^k = \sum_{ij}\m{I}^k_{ij}$ has eigenvalues that count the number of edges in the $k$th one-particle state.
   \begin{align}
       \m{N}^k\ket{G_0,\dots,G_{D-1}}&=\sum_{ij}\m{I}^k_{ij}\ket{G_0,\dots,G_{D-1}}=n_k\ket{G_0,\dots,G_{D-1}},
   \end{align} 
   where the eigenvalue $n_k$ is the number of edges in $G_k$.
 \end{defn}
\begin{defn}
    Let $g(V,E)$ be a multigraph with vertex set $V\subseteq [N]$ and edge set $E=\{e_1,e_2,\dots,e_n\}$, the {\bf subgraph raising and lowering} operators $L^+_g, L^-_g$ are defined as 
    \begin{align}
        L^+_g\equiv \prod_{e\in E}L^+_e,\quad L^-_g\equiv \prod_{e\in E}L^-_e.
    \end{align}
\end{defn}

\begin{defn}
    Let $g(V,E)$ be a multigraph with vertex set $V\subseteq [N]$ and edge set $E=\{e_1,e_2,\dots,e_n\}$, the {\bf subgraph indicator} operator $\m{I}^k_g$ is defined as 
    \begin{align}
        \m{I}^k_g\equiv \prod_{e\in E}\m{I}^k_e,
    \end{align}
    where $\m{I}^k_e$ are the one-particle indicator operators defined in Eq.(\ref{eq:oneParticleNumber}).    
Using the commutation relations of the ladder operators we can also write 
\begin{align}
    \m{I}^k_g=\left(L^+_g\right)^k\left(L^-_g\right)^{D-1}\left(L^+_g\right)^{D-1-k}
\end{align}
\end{defn}
\begin{defn}
Let $g$ be an unlabeled multigraph with $m$ vertices where $m\le N$. There are $l(g)\equiv\binom{N}{m}\frac{m!}{|\Gamma(g)|}$ ways to label the vertices of $g$ using $m$ labels chosen from the set $[N]$ (assume $m\le N$), where each labeled multigraph is distinct as vertex-labeled multigraph. Let $\m{L}(g)=\{g_1,g_2,\dots g_{l(g)}\}$ be the {\em labelings} of $g$, i.e., the set of all vertex-labeled multigraphs $g_i$ that are each isomorphic to $g$. We define the {\bf subgraph occupation number} operator $\m{N}^k_g$ as the sum of the subgraph indicator operators over all labelings,
\begin{align}
    \m{N}^k_g = \sum_{l\in \m{L}(g)}\m{I}^k_l
\end{align}
When it acts on a quantum graph state $\ket{G_0,\dots,G_{D-1}}$, this operator counts the number $n^{g}_k$ of distinct labeled subgraphs in $G_k$ that are isomorphic to $g$.
   \begin{align}
       \m{N}^k_g\ket{G_0,\dots,G_{D-1}}&=n^g_k\ket{G_0,\dots,G_{D-1}},
   \end{align} 
   where the eigenvalue is $n^g_k$.
\end{defn}
In particular, if $g=K_m$, the complete graph with $m$ vertices ($m\le N$), then $|\Gamma(K_m)|=m!$ so there are $\binom{N}{m}$ ways to label the vertices of $K_m$ using labels in $[N]$, and 
\begin{align}
    \m{N}^k_{K_m}=\sum_{\{v_1,\dots,v_m\}\subseteq [N]}\m{I}^k_{K_{\{v_1,\dots,v_m\}}},
\end{align}
where $K_{\{v_1,\dots,v_m\}}$ is the complete graph with vertex set $\{v_1,\dots,v_m\}$.

\section{Unlabeled Quantum Multigraphs}
\label{sec:unlabeledQuantumGraphs}
In this section we construct the Hilbert space of unlabeled quantum multigraphs with $N$ vertices. In many combinatorial problems we are not interested in making a distinction between two isomorphic graphs, i.e., if the only distinction between two graphs is in the labeling of their vertices and not in the underlying adjacency among them, then we do not want to consider them as distinct. This is also the case when studying physical systems made up of identical constituents. This arises in many-body quantum mechanics of identical particles. Although the total Hilbert space is constructed by first assigning arbitrary labels to the particles and treating them as distinguishable, we have to impose proper symmetrization or antisymmetrization procedures (depending on the statistics of the particles) to remove the distinction resulting from the arbitrary labeling. Physically, this amounts to the requirement or restriction that the physical states are the same (up to an overall phase) when we relabel the arbitrary labels assigned to the identical constituents. Mathematically, this means that the physical states transform through simple multiplication by a phase under the action of the symmetric group. This is achieved by projecting the tensor product space down to the 1D invariant subspaces under the action of the symmetric group.  Symmetrization and antisymmetrization procedures are exactly the projections down to the two 1D invariant subspaces. Similar construction can be carried out on quantum multigraph Hilbert spaces.

Before dealing with unlabeled quantum multigraphs, it is instructive to first review the construction of the physical Hilbert space of a system of identical $N$-particles in ordinary many-body quantum mechanics using the occupation-set basis. We start with the assumption that the one-particle  Hilbert space $\m{H}=\mathrm{Span}\{\ket{m}, m=0,1,\dots,D-1\}$ is known. To construct the physical Hilbert space of the $N$ identical particles, we start by treating these $N$ particles as distinguishable by arbitrary assigning them labels (e.g. $1,2,\dots,N$). Now their respective Hilbert spaces are $\m{H}_1, \m{H}_2,\dots,\m{H}_N$. The total Hilbert space of the $N$ ``distinguishable" particle system $\m{H}_{tot}$ is the tensor product $\m{H}_{tot}=\otimes_{i=1}^N\m{H}_i=\mathrm{Span}\{\ket{m_1}_1\ket{m_2}_2\dots\ket{m_n}_n\}$. We define the action of the symmetric group $S_N$ on $\m{H}_{tot}$ given by permutation of the particle labels (without changing the one-particle Hilbert space labels $m_k\in\{0,\dots,D-1\}$) followed by sorting into lexicographic order. So, for example, for $\pi=(12)\in S_N$,

\begin{align} \pi\big(\ket{m_1}_1\ket{m_2}_2\dots\ket{m_n}_n\big) 
&=\ket{m_1}_{\pi(1)}\dots\ket{m_n}_{\pi(n)}=\ket{m_1}_2\ket{m_2}_1\dots\ket{m_n}_n\CR
&=\ket{m_2}_1\ket{m_1}_2\dots\ket{m_n}_n
\end{align}
This action of $S_N$ on $\m{H}_{tot}$ simply permutes the basis states, and therefore results in a $D^N\times D^N$ dimensional permutation representation of $S_N$. 

Analogous to the occupation graph basis, we can use ``occupation set" basis where the tensor product ket $\ket{m_1}_1\ket{m_2}_2\dots\ket{m_n}_n$ is assigned to $\ket{S_0,S_1,\dots,S_{D-1}}$, with $(S_0,\dots,S_{D-1})$ a weak ordered set partition of $[N]$ into $D$ parts and $S_i$ made up of all particles in the $\ket{i}$ one-particle state. The difference is that now we are considering partitions of the set $[N]$ as opposed to partitions of $[N]\times[N]$ or $\binom{[N]}{2}$. A permutation $\pi\in S_N$ acts on an occupation set basis ket $\ket{S_0,S_1,\dots,S_{D-1}}$ as $ \pi\ket{S_0,S_1,\dots,S_{D-1}}=\ket{\pi(S_0),\pi(S_1),\dots,\pi(S_{D-1})}$. This action does not change the sizes of the parts $|S_i| = |\pi(S_i)|$, and so unique orbits under $S_N$ are characterized by unique $D$-tuples $(n_0,n_1,\dots,n_{D-1})$.

Now we demand that the physical states of the identical particle system are those that remain the same (up to multiplication by a phase) when acted on by the permutation group. For $N\ge2$ there are only two 1D irreducible representations of the symmetric group, the trivial representation and the sign representation which assigns each permutation $\pi\in S_N\rightarrow \mathrm{sgn}(\pi)\in\{\pm 1\}$. Therefore, we need to find the irreducible invariant subspaces of the tensor product Hilbert space on which the permutation group acts either trivially or under the sign action. That means that the physical states $\ket{\psi}\in\m{H}_{physical}$ have to obey $\pi\ket{\psi}=\ket{\psi}$ or $\pi\ket{\psi}=\sgn(\pi)\ket{\psi}$ for all $\pi\in S_N$. 

These irreducible invariant subspaces are found via the symmetrizer ($\m{S}$) and antisymmetrizer ($\m{A}$) projection operators which project $\m{H}_{tot}$ to $\m{H}_{physical}$ and are given, for any basis ket in $\psi \in\m{H}_{tot}$ by 
 \begin{align}
    \m{S}\ket{\psi}&=\sqrt{\frac{1}{N!}}\sum_{\pi\in S_N}\pi\ket{\psi}\CR
     \m{A}\ket{\psi}&=\sqrt{\frac{1}{N!}}\sum_{\pi\in S_N}\mathrm{sgn}(\pi)\pi\ket{\psi}
 \end{align}
The normalization factor $\sqrt{1/N!}$ keeps the normalized states normalized after the projection. 

We can see that the unique physical states resulting from the (anti)symmetrization projection are the unique orbits characterized by a $D$-tuple of numbers $(n_0,\dots,n_{D-1})$ corresponding to the number of particles in each single particle state. Each unique physical state is fully described by listing the number of particles in the single particle state $\ket{k}$ for each $k$. So we can move from the occupation set basis to the familiar occupation number basis $\ket{n_0,\dots,n_{D-1}}$. The number of symmetric basis kets of the symmetric physical Hilbert space is given by the number of weak compositions of $N$ into $D$ parts, or $\binom{N+D-1}{N}$. On the other hand, for $D\ge N$, there are $\binom{D}{N}$ states that can be formed which are fully antisymmetric by superimposing occupation sets with all but $\binom{D}{N}$ of the parts $S_k$ empty and the rest having exactly one element. For example, for $D=3, N=2,$ there are three fully antisymmetric kets given by 
\begin{align}
\ket{n_0=0,n_1=1,n_2=1}&=\m{A}\ket{\{\},\{1\},\{2\}},\CR
\ket{n_0=1,n_1=0,n_2=1}&=\m{A}\ket{\{1\},\{\},\{2\}},\CR
\ket{n_0=1,n_1=1,n_2=0}&=\m{A}\ket{\{1\},\{2\},\{\}}\nonumber.
\end{align}
The Pauli exclusion principle is a natural consequence of the antisymmetrization operator which annihilates any basis kets in $\m{H}_{tot}$ if any two or more particles have the same single-particle quantum number, i.e., if any of the parts $S_0,\dots,S_{D-1}$ have more than one element. The principle of symmetrization in quantum mechanics postulates that if the individual particles are Bosons, then we have to symmetrize and if they are Fermions, we have to antisymmetrize. This becomes the Spin-Statistics theorem proven in the context of quantum field theory in $3+1$ dimensions \cite{streater2000pct}. However, outside of the specific setting of $3+1$ dimensional QFT there is no general theorem that ties statistics (symmetrization of antisymmetrization) with spin.

Now we return to quantum graphs. For a fixed number of vertices $N$, the action of the permutation group on the labeled quantum multigraph states provides a $\dim \m{H}\times \dim \m{H}$ dimensional permutation representation of $S_N$. To find the Hilbert space of unlabeled quantum multigraphs, we impose the restriction that the action of the symmetric group on any ``physical" unlabeled quantum multigraph be simply multiplication by a phase, leading to two possibilities for each $1D$ irreducible representation of $S_N$:
\begin{align}
\pi(\ket{G}^S_{unlabeled})&=\ket{G}^S_{unlabeled}\CR
\pi(\ket{G}^A_{unlabeled})&=\mathrm{sgn}(\pi)\ket{G}^A_{unlabeled}
\end{align}

The action of $S_N$ on the quantum multigraphs is such that the relabeling of vertices in any given graph does not change the adjacency structure of the graphs. Therefore, $G$ and $\pi(G)$ are isomorphic graphs for any graph $G$ and any $\pi\in S_N$. As a result, the $\dim\m{H}\times \dim\m{H}$ dimensional representation of $S_N$ on $\m{H}$ will be block diagonal with each block corresponding to a multigraph isomorphism class. However, the blocks themselves are still not irreducible. Each block has a one-dimensional irreducible subspace constructed as follows. 

Given a labeled quantum graph $\ket{G}=\ket{G_0,\dots,G_{D-1}}$, its orbit under the action of $S_N$, $\mathrm{orb}(\ket{G})=\{\pi\ket{G}: \pi\in S_N\}$ has $N!/|\Gamma(\ket{G})|$ unique (up to multiplication by a phase) labeled elements. That is because $S_N$ can be partitioned into the left cosets $S_N/\Gamma(\ket{G})=\{\pi\Gamma(\ket{G}):\pi\in S_N\}$ as $S_N=t_1\Gamma(\ket{G})\cup t_2\Gamma(\ket{G})\cup\dots\cup t_k\Gamma(\ket{G})$, where $t_1,\dots,t_k\in S_N$ are a set of unique coset representatives, and if two permutations $\pi_1,\pi_2$ are in the same coset, then $\pi_1\ket{G}=\pm\pi_2\ket{G}$ since $\pi_1=t_i\sigma_1, \pi_2=t_i\sigma_2$, with $\sigma_1,\sigma_2\in\Gamma(\ket{G})$, and $t_i$ the coset representative, and 
\begin{align}
    \pi_1\ket{G}&=t_i\sigma_1\ket{G}=\mathrm{sgn}(\sigma_1)t_i\ket{G}\CR
    \pi_2\ket{G}&=t_i\sigma_2\ket{G}=\mathrm{sgn}(\sigma_2)t_i\ket{G}.
\end{align}
So, the action of $S_N$ on the orbit $\mathrm{orb}(\ket{G})$ is itself either a permutation representation (called the coset representation) if the Hilbert space is $\m{H}^S$ or a tensor product of the coset representation with the sign representation if the Hilbert space is $\m{H}^A$. Either way, since permutation representations are not irreducible but instead always have the trivial representation as an irreducible subspace, it follows that within each block corresponding to a graph isomorphism class we find a one-dimensional irreducible representation that is either the trivial representation (for $\m{H}^S$) or the sign representation (for $\m{H}^A$). The physical Hilbert space of unlabeled quantum graphs is spanned by these 1D irreducible subspaces. 

The projection operators from the labeled quantum graphs down to the 1D irreducible representations of the unlabeled Hilbert space are the symmetrization $\m{S}$ and antisymmetrization $\m{A}$ rules defined as 
\begin{align}
    \m{S}\ket{G}&=\frac{1}{N!}\sum_{\pi\in S_N}\pi\ket{G}\CR
    \m{A}\ket{G}&=\frac{1}{N!}\sum_{\pi\in S_N}\mathrm{sgn}(\pi)\,\pi\ket{G}
\end{align}
Note that these operators have a different normalization compared to the symmetrizer and antisymmetrizer operators in many-body quantum mechanics. The reason for this choice of normalization is to make the operators idempotent. 
The projection operators can also written as 
\begin{align}
    \m{S}\ket{G}&=\frac{|\Gamma(\ket{G})|}{N!}\sum_{t\in T}t\ket{G}\CR
    \m{A}\ket{G}&=\frac{|\Gamma(\ket{G})|}{N!}\sum_{t\in T}\mathrm{sgn}(t)\,t\ket{G},
\end{align}
where $T=\{t_1,t_2,\dots,t_k\}$ is a set of coset representatives. We defined the unlabaled quantum multigraphs to be the results of these projection operators. 

As an example, consider the labeled antisymmetric quantum multigraph $\ket{G}^A=\ket{\{12,14,23,34\},\{\},\{13,24\}}^A$ shown diagrammatically in Fig.\ref{fig:antisymmetrization}. Its automorphism group is the dihedral group on four vertices, \[\Gamma(\ket{G})=D_4=\{e,(13),(24),(1234),(1432),(12)(34),(13)(24),(14)(23)\},\] where $e$ is the identity permutation. The left coset has three elements with representatives given by $T(S_4/\Gamma(\ket{G}))=\{e,(12),(14)\}$
So the antisymmetrizer gives
\begin{align}
    \m{A}\ket{G}^A=\frac{1}{3}\Big(&\ket{\{12,14,23,34\},\{\},\{13,24\}}^A-(12)\ket{\{12,14,23,34\},\{\},\{13,24\}}^A\CR
    &-(14)\ket{\{12,14,23,34\},\{\},\{13,24\}}^A\Big)\CR
    =\frac{1}{3}\Big(&\ket{\{12,14,23,34\},\{\},\{13,24\}}^A-(-1)^{1}\ket{\{12,24,13,34\},\{\},\{23,14\}}^A\CR
    &-(-1)^{5}\ket{\{24,14,23,13\},\{\},\{34,12\}}^A\Big)\CR
    =\frac{1}{3}\Big(&\ket{\{12,14,23,34\},\{\},\{13,24\}}^A+\ket{\{12,24,13,34\},\{\},\{23,14\}}^A\CR
    &+\ket{\{24,14,23,13\},\{\},\{34,12\}}^A\Big)
\end{align}
The antisymmetrized state corresponding to the unlabeled quantum multigraph is shown in Fig.\ref{fig:antisymmetrization}.
\begin{figure*}[h!]
\centering
    \includegraphics[width=0.7\textwidth]{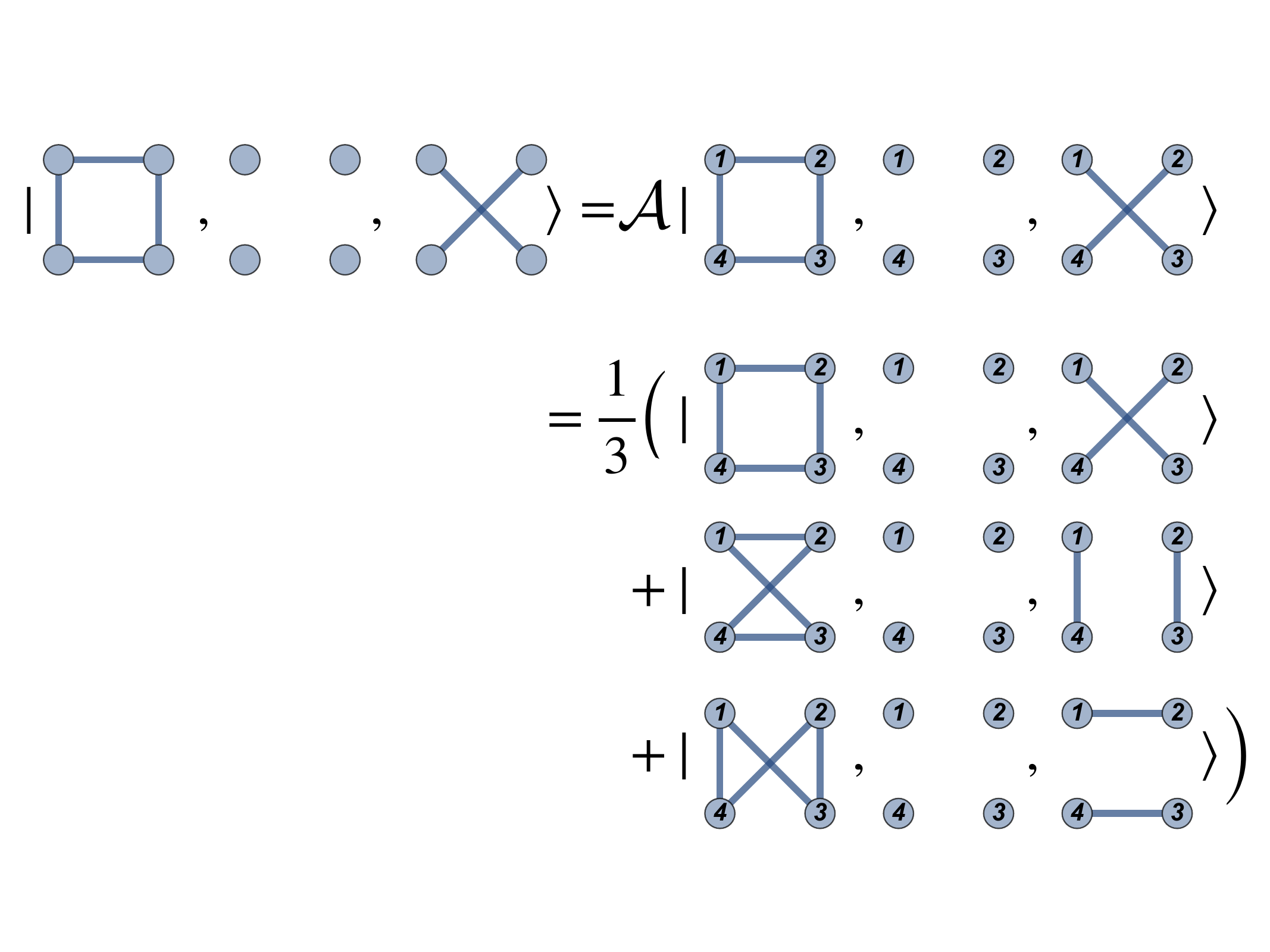}  
\caption{A diagrammatic representation of the unlabeled quantum multigraph given in the occupation graph basis as $\ket{G}^A=\m{A}\ket{\{12,14,23,34\},\{\},\{13,24\}}$ with $N=4$ and $\m{H}_{ij}=\mathrm{Span}\{\ket{0},\ket{1},\ket{2}\}$.}
\label{fig:antisymmetrization}
\end{figure*}
\subsection{Properties of the (anti)symmetrizer}
The projection under the (anti)symmetrizer satisfies all the requirements we have set for the physical Hilbert space of unlabeled quantum graphs. For any $\pi\in S_N$ the following hold: 
\begin{lem}
For any $\pi\in S_N$, 
    \begin{align}
      \pi\m{S}\ket{G}^S=\m{S}\pi\ket{G}^S=\m{S}\ket{G}^S\\
      \pi\m{A}\ket{G}^A=\m{A}\pi\ket{G}^A=\mathrm{sgn}(\pi)\m{A}\ket{G}^A
    \end{align}    
\end{lem}

\begin{proof}
    \begin{align}
    \pi\m{S}\ket{G}^S &=\frac{1}{N!}\sum_{\sigma\in S_N}\pi\sigma\ket{G}^S=\frac{1}{N!}\sum_{\sigma'\in S_N}\sigma'\ket{G}^S,\;\text{ where }\sigma'=\pi\sigma\CR
    &=\m{S}\ket{G}^S\CR
    \m{S}\pi\ket{G}^S &=\frac{1}{N!}\sum_{\sigma\in S_N}\sigma\pi\ket{G}^S=\frac{1}{N!}\sum_{\sigma''\in S_N}\sigma''\ket{G}^S,\;\text{ where }\sigma''=\sigma\pi\CR
    &=\m{S}\ket{G}^S\nonumber
    \end{align}
    \begin{align}
    \pi\m{A}\ket{G}^A&=\frac{1}{N!}\sum_{\sigma\in S_N}\mathrm{sgn}(\sigma)\pi\sigma\ket{G}^A=\frac{1}{N!}\sum_{\sigma\in S_N}\mathrm{sgn}(\pi)\mathrm{sgn}(\pi\sigma)\pi\sigma\ket{G}^A\CR
    &=\mathrm{sgn}(\pi)\frac{1}{N!}\sum_{\sigma'\in S_N}\mathrm{sgn}(\sigma')\sigma'\ket{G}^A\;\text{ where }\sigma'=\pi\sigma\CR
    &=\mathrm{sgn}(\pi)\m{A}\ket{G}^A\CR
    \m{A}\pi\ket{G}^A&=\frac{1}{N!}\sum_{\sigma\in S_N}\mathrm{sgn}(\sigma)\sigma\pi\ket{G}^A=\frac{1}{N!}\sum_{\sigma\in S_N}\mathrm{sgn}(\pi)\mathrm{sgn}(\sigma\pi)\sigma\pi\ket{G}^A\CR
    &=\mathrm{sgn}(\pi)\frac{1}{N!}\sum_{\sigma''\in S_N}\mathrm{sgn}(\sigma'')\sigma''\ket{G}^A\;\text{ where }\sigma''=\sigma\pi\CR
    &=\mathrm{sgn}(\pi)\m{A}\ket{G}^A\nonumber    
    \end{align}
where we have used the property of the $\mathrm{sgn}$ function: $\mathrm{sgn}(\pi_1\pi_2)=\mathrm{sgn}(\pi_1)\mathrm{sgn}(\pi_2)$.
\end{proof}

\begin{lem}
Symmetrizing the antisymmetric kets or antisymmetrizing the symmetric kets always results in their annihilation, i.e., $\m{A}\ket{G}^S = 0$, and $\m{S}\ket{G}^A = 0$.
\end{lem}
\begin{proof}
\begin{align}   
\m{A}\m{S}\ket{G}&=\frac{1}{N!}\sum_{\pi\in S_N}\sgn(\pi)\pi\m{S}\ket{G}\CR
&=\frac{1}{N!}\left(\sum_{\pi\in S_N}\sgn(\pi)\right)\m{S}\ket{G}=0\CR
\m{S}\m{A}\ket{G}&=\frac{1}{N!}\sum_{\pi\in S_N}\pi\m{A}\ket{G}\CR
&=\frac{1}{N!}\sum_{\pi\in S_N}\sgn(\pi)\m{A}\ket{G}\CR
&=\frac{1}{N!}\left(\sum_{\pi\in S_N}\sgn(\pi)\right)\m{A}\ket{G}=0\nonumber
\end{align}
\end{proof}

\begin{lem}
Both the symmetrizer and antisymmetrizer are projection operators (they are idempotent,  $\m{S}\m{S}=\m{S}$, and $\m{A}\m{A}=\m{A}$). Note, due to the normalization we chose, a normalized labeled graph basis state doesn't stay normalized after projection; i.e., for a labeled state $\ket{G}^A=\ket{G_0,\dots,G_{D-1}}^A$, if $^A\braket{G}{G}^A=1$, then $^A\braket{\m{A}G}{\m{A}G}^A=1/N!$.
\end{lem}
\begin{proof}
\begin{align}
    \m{S}\m{S}\ket{G}^S&=\frac{1}{N!}\sum_{\pi\in S_N}\pi\m{S}\ket{G}^S=\frac{1}{N!}\sum_{\pi\in S_N}\m{S}\ket{G}^S\CR
    &=\frac{1}{N!}N!\m{S}\ket{G}^S=\ket{G}^S\nonumber
\end{align}        
\begin{align}
    \m{A}\m{A}\ket{G}^A&=\frac{1}{N!}\sum_{\pi\in S_N}\sgn(\pi)\pi\m{A}\ket{G}^A\CR
    &=\frac{1}{N!}\sum_{\pi\in S_N}\sgn(\pi)\sgn(\pi)\m{A}\ket{G}^A\CR
    &=\frac{1}{N!}N!\m{A}\ket{G}^A=\m{A}\ket{G}^A\nonumber
\end{align}        
\end{proof}

Analogous to the way we represented labeled quantum multigraph states $\ket{G}=\ket{G_0,\dots,G_{D-1}}$, where the $G_i$ are labeled graphs in the single particle state $i$ and $(G_0,\dots,G_{D-1})$ is a weak ordered set partition of the edge space, we can similarly represent unlabeled quantum graphs as 
$\ket{G}^A_u\equiv\m{A}\ket{G}^A$ or $\ket{G}^S_u\equiv \m{S}\ket{G}^S$ as $\ket{G}_u=\ket{G_0^u,\dots,G_{D-1}^u}$, where each $G_i^u$ is unlabeled graph and $(G_0^u,\dots,G_{D-1}^u)$ is a weak ordered partition of the complete unlabeled graph on $N$ vertices.

\subsection{Unlabeled Operators}
All operators in the labeled quantum graph Hilbert space have counterparts that act on the unlabeled Hilbert space found through the (anti)symmetrization projections. The unlabeled operators in the antisymmetric quantum multigraph Hilbert space can be defined through their matrix elements as
\begin{align}
    \bra{G}\m{A}\m{O}_\ell \m{A}\ket{G'}&=\bra{G}\m{A}(\m{A}\m{O}_\ell \m{A})\m{A}\ket{G'}\CR
    \m{O}_u&=\m{A}\m{O}_\ell \m{A}
\end{align}
For example, the ladder operators $L^\pm_{ij}$ in the labeled quantum graph Hilbert space have counterparts in the antisymmetric Hilbert space given by
\begin{align}
    L^\pm &=\m{A}L^\pm_{ij}\m{A}=\frac{1}{N!}\sum_{\pi\in S_N}\pi L^\pm_{ij}\pi\CR
    &=\frac{1}{N!}\sum_{\pi\in S_N}L^\pm_{\left(\pi(i)\pi(j)\right)}\CR
    &=\frac{1}{N!}\left(L^\pm_{12}+L^\pm_{13}+\dots+L^\pm_{N-1,N}\right),
\end{align}
Similarly, the subgraph ladder operators project to
\begin{align}
    L_{g_u}^\pm &=\m{A}L^\pm_{g_l}\m{A}=\frac{1}{N!}\sum_{\pi\in S_N}\pi L^\pm_{g_l}\pi\CR
    &=\frac{1}{N!}\sum_{\pi\in S_N}L^\pm_{\pi(g_l)}\CR
    &=\frac{|\Gamma(g_l)|}{N!}\sum_{g'}L^\pm_{g'},
\end{align}
where $g_u$ is an unlabeled graph and $g_l$ is a labeling of $g_u$. In the last line, the sum is done over all labeled graphs $g'$ that are isomorphic to $g_u$.

One consequence of this is that there are no ``local" vertex operators. More precisely, any vertex operator $\m{O}^k_i$ acting on the (anti)symmetrized states produces the average value of $\langle \m{O}^k\rangle = (1/N)\sum_i\m{O}^k_i$. For example, consider the degree operator $\deg^k_i$ which counts the degree of vertex $i$ in the $k$th one-particle state graph $G_k$ defined as
\begin{align}
    \deg^k_i=\sum_{j}\m{I}^k_{ij},
\end{align}
the projection of $\deg^k_i$ to the unlabeled Hilbert space is 
\begin{align}
    \deg^k&=\m{A}\deg^k_i\m{A}=\frac{1}{N!}\sum_{\pi\in S_N}\deg^k_{\pi(i)}\CR
    &=\frac{1}{N}\sum_{i=1}^N\deg^k_i
\end{align}
For illustration, consider the labeled antisymmetric graph state\\
$\ket{G}^A=\ket{\{12,23,24\},\{\},\{13,14,34\}}^A$ and the unlabeled antisymmetric graph state $\m{A}\ket{G}^A$ shown in Fig.\ref{fig:nolocal}, then 
\begin{align}
    \deg^0_2\ket{G}&=3\ket{G};\,\deg^0_1\ket{G}=\deg^0_3\ket{G}=\deg^0_4\ket{G}=1\ket{G},\text{ and}\CR
    \deg^0\left(\m{A}\ket{G}^A\right)&=\frac{3}{2}\left(\m{A}\ket{G}^A\right)
\end{align}

\paragraph*{Remark on quantum symmetries.}
The projection to the unlabeled Hilbert space implements invariance under the classical permutation group $S_N$. Foundationally, this sits within the theory of \emph{quantum} automorphism groups of finite sets and graphs \cite{Wang1998CMP,Bichon2003PAMS}, which generalize $S_N$ to compact matrix quantum groups and have seen renewed interest for multigraphs \cite{goswami2023quantum,goswami2024quantum}. Our construction here only requires the classical action of $S_N$, but the operator‑algebraic perspective clarifies why only permutation‑invariant combinations of vertex observables survive the projection.

\begin{figure}
\centering
    \includegraphics[width=0.7\textwidth]{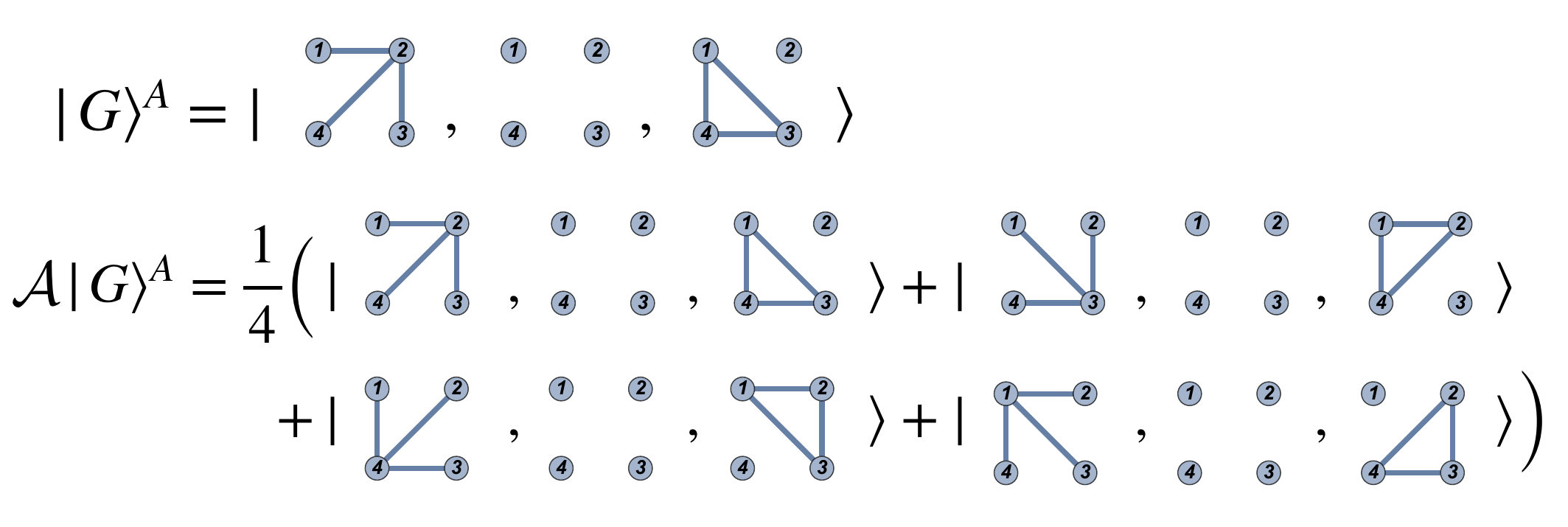}  
\caption{A labeled graph $\ket{G}^A$ (top row) and its unlabeled version $\m{A}\ket{G}^A$ (bottom row). Local vertex observables in the labeled Hilbert space will be projected to global operators in the unlabeled Hilbert space that correspond to taking averages in the labeled Hilbert space. }
\label{fig:nolocal}
\end{figure}

\section{Dynamical Quantum Graphs}
\label{sec:dynamics}
We now introduce a general dynamics on quantum multigraphs. After formulating a general Hamiltonian that is diagonal in the occupation graph basis, we focus specifically on quantum graphs (i.e., $D=2$) and study the thermodynamics of two simple Hamiltonians. The first is the non-interacting or free Hamiltonian $(H_0)$ and the second is a Hamiltonian of nearest-neighbor Ising type interaction $(H_{Ising})$. The thermodynamics of these Hamiltonians is studied numerically through Monte Carlo simulations for both labeled and unlabeled quantum graphs. The main results of this section are the following. We give analytic and simulation results that show:
\begin{enumerate}
    \item The free labeled quantum graph system effectively becomes the Erd\H{o}s--R\'enyi--Gilbert random graph model; it has no thermodynamic phase transitions.
    \item The Ising model of the labeled quantum graphs does not have thermodynamic phase transitions.
    \item The unlabeled quantum graphs, both the free system and the ferromagnetic Ising system, appear to have proper thermodynamic phase transitions marked by diverging specific heat and critical slowing down near the phase transition. This phase transition is associated with the emergence of a connected component in the excited state that includes almost all vertices.
    \item The unlabeled Ising antiferromagnet has no phase transition. In addition, the labeled and unlabeled antiferromagnets have converging thermodynamic functions.
\end{enumerate}

To introduce dynamics to quantum multigraphs, first consider the case where the one-particle Hilbert space $\m{H}=\mathrm{Span}\{\ket{0}\dots,\ket{D-1}\}$ is spanned by eigenvectors of some one-particle Hamiltonian $H$, with $H\ket{n}=E_n\ket{n}$. Without loss of generality, we assume that the one-particle eigenstates $\ket{n}$ are ordered in a nondecreasing order, i.e., $E_n\le E_{n+1}$ for all $0\le n\le D-1$. Such a one-particle Hamiltonian extends to a Hamiltonian on the quantum multigraph Hilbert space in the non-interacting system as
\begin{align}
    H_0&=\sum_{ij}H_{ij}, \;\text{where}\CR
    H_{ij}&=I_{11}\otimes I_{12}\otimes\dots\otimes H_{ij}\otimes\dots\otimes I_{NN},
\end{align}
where by abuse of notation we use $H_{ij}$ to now represent the operator on the quantum multigraph Hilbert space. The energy of the quantum graph $\ket{G_0,\dots,G_{D-1}}$ is given by 
\begin{align}
    H_0\ket{G_0,\dots,G_{D-1}}&=\left(\sum_{k=0}^{D-1}E_k\m{N}^k\right)\ket{G_0,\dots,G_{D-1}}\CR
    &=\left(\sum_{k=0}^{D-1}E_kn_k\right)\ket{G_0,\dots,G_{D-1}},
\end{align}
where $n_k$, the number of edges in $G_k$, is the eigenvalue of the edge occupation number operator $\m{N}^k$. 

We can define a general interacting Hamiltonian that is diagonal in the occupation graph basis as follows. First, for a given graph isomorphism class $g$, define a ``graph energy" operator $H_g$ as 
\begin{align}
    H_g=\sum_{k=0}^{D-1} E^k_g\m{N}^k_g,
\end{align}
where $\m{N}^k_g$ counts the number of subgraphs in the $k$th one-particle state that are isomorphic to $g$ and $E^k_g$ gives the contribution to the total energy of this graph. Then, the most general Hamiltonian that is diagonal in the occupation graph basis is given by summing $H_g$ over all graph isomorphism classes $g$ with some arbitrary weight $f(g)$:
\begin{align}
    H=\sum_g f(g)\sum_{k=0}^{D-1}E^k_g\m{N}^k_g,
\end{align}
The free Hamiltonian $H_0$ is a special case where $f(g)$ is 0 unless $g$ is an edge and $E^k_g=E_k$. 

We consider a simple interacting Hamiltonian that we will term the ``Ising" Hamiltonian $H_{Ising}$. The Ising Hamiltonian is inspired by the nearest-neighbor interaction in the Ising model \cite{PhysRev.65.117, feynman1998statistical}. The nearest neighbors of the edge $\{v_1,v_2\}$ in the complete graph are all edges of the form $\{v_i,v_1\}$ and $\{v_2,v_j\}$. Then, the Ising Hamiltonian is given by 
\begin{align}
    H_{Ising}=\sum_{k=0}^{D-1}E_k\sum_{\{i,j,\ell\}\subset[N]}\left(\m{I}^k_{ij}\m{I}^k_{i\ell}+\m{I}^k_{ij}\m{I}^k_{j\ell}+\m{I}^k_{i\ell}\m{I}^k_{j\ell}\right)
\end{align}
\begin{figure}
\centering
    \includegraphics[width=0.7\textwidth]{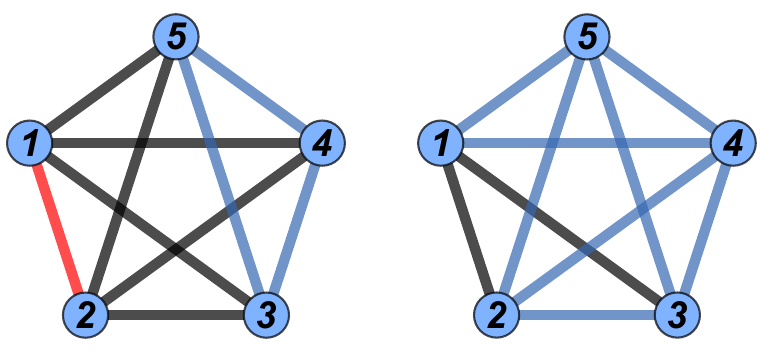}  
\caption{(Left) The edge $\{1,2\}$ (highlighted in red), its six neighboring edges (in black), and three edges that are not adjacent to it (in blue). (Right) Highlighted in black is an angle subgraph terminating at the vertices $2$ and $3$.}
\label{fig:IsinghamiltonianIllustration}
\end{figure}
For example, for the simple case of $N=3$, the Ising Hamiltonian is given by 
\begin{align}
H_{Ising}&=\sum_{k=0}^{D-1}\left( \m{I}^k_{12}\m{I}^k_{13}+\m{I}^k_{12}\m{I}^k_{23}+\m{I}^k_{13}\m{I}^k_{23}\right)
\end{align}
The operator $\sum_{\{i,j,\ell\}\subset[N]}\left(\m{I}^k_{ij}\m{I}^k_{i\ell}+\m{I}^k_{ij}\m{I}^k_{j\ell}+\m{I}^k_{i\ell}\m{I}^k_{j\ell}\right)$ basically counts the number of subgraphs in the $k$th one-particle state isomorphic to the ``angle" graph (see Fig.\ref{fig:IsinghamiltonianIllustration}), i.e., two vertices both connected to a third vertex. So, the Ising Hamiltonian can be written as
\begin{align}
\label{eq:IsingHamiltonian}
    H_{Ising}=\sum_{k=0}^{D-1}E_k\left( \frac{1}{2}\sum_{i\neq j}b^k_{ij}\right) = \sum_{k=0}^{D-1}E_kb^k,
\end{align}
where $b_{ij}^k$ is the number of 2-walks between vertices $i$ and $j$ in graph $G_k$ and $b^k$ is the total number of these angle subgraphs in the $k$th graph. Since the number of $m$-walks between vertices $i$ and $j$ is given in terms of the adjacency matrix $A$ as $(A^m)_{ij}$, the Ising Hamiltonian can be expressed in terms of the adjacency matrices of the graphs $G_0, G_1,\dots, G_{D-1}$. Using the same symbols $G_k$ to represent the adjacency matrix of graph $G_k$, we can write the Hamiltonian as 
\begin{align}
    H_{Ising}=\frac{1}{2}\sum_{k=0}^{D-1}E_k \sum_{i\neq j}(G_k)^2_{ij},
\end{align}
Alternatively, a vertex with degree $d$ has $\binom{d}{2}$ pairs of vertices that are connected to it. Therefore, 
\begin{align}
\label{eq:HIsingDegreeVsAngleCount}
    H_{Ising}&=\sum_{k=0}^{D-1}E_k \sum_{i=1}^N \binom{d^k_i}{2},
\end{align}
where $d^k_i$ is the degree of vertex $i$ in the $k$th graph $G_k$.

We are interested in the thermodynamic properties of the dynamical quantum multigraphs in the canonical ensemble for the two Hamiltonians (free and Ising) introduced above. We will restrict our focus to quantum graphs, i.e., the case where the one-particle Hilbert space is 2-dimensional, so that the occupation set basis kets are of the form $\ket{\{G_0,G_1\}}$, where $G_1$ is the graph complement of $G_0$. A graph and its complement have the same automorphism group; therefore, after antisymmetrizing the unlabeled quantum graph basis kets are enumerated simply by unlabeled graphs. 

We are interested in, (1). the zero temperature or ground state(s), (2). phase transitions, and (3). the typical graph in each phase. We will study the dynamics in both labeled and unlabeled systems in order to see the effect of vertex labeling on the thermodynamics. Analytic computations are difficult for all cases except for the free labeled system; therefore, we will rely on Monte Carlo simulations. 

In the canonical ensemble, at inverse temperature $\beta$, the probability of finding the system in a particular quantum graph state $\ket{G}$ is given by its Boltzmann weight 
\begin{align}
    P(\ket{G})&=\frac{e^{-\beta E(G)}}{Z_N},
\end{align}
where $Z_N$ is the partition function for the system with $N$ vertices. For the labeled quantum graph system, the partition function $Z^l_N$ is 
\begin{align}
    Z^l_N&=\sum_{G_l}e^{-\beta E(G_l)},
\end{align}
where $G_l$ runs over all labeled graphs with $N$ vertices. The unlabeled quantum graphs partition function $Z^u_N$ is 
\begin{align}
    Z^u_N&=\sum_{G_u}e^{-\beta E(G_u)}\CR
    &=\sum_{G_l}\frac{|\Gamma(G_l)|}{N!}e^{-\beta E(G_l)},
\end{align}
where $G_u$ runs over all unlabeled graphs with $N$ vertices and $|\Gamma(G_l)|$ is the size of the automorphism group of $G_l$ defined in \ref{eq:automorphismGroup}. In the second line, we are summing over all labeled graphs with an additional factor of $|\Gamma(G_l)|/N!$ which accounts for the overcounting of the same unlabeled graph when using labeled graphs.

Numerical simulations are done using the Metropolis Monte Carlo algorithm where given the current graph $G$, a single edge flip is attempted and the new graph with the flipped edge ($G'$) is accepted or rejected with an acceptance probability given by 
\begin{align}
\label{eq:sampling}
    A(G\rightarrow G')&=\mathrm{Min}\left\{1,e^{-\beta(E(G')-E(G))} \right\}\;\;\text{ for labeled graphs,}\nonumber\\
    A(G\rightarrow G')&=\mathrm{Min}\left\{1,\frac{|\Gamma(G')|}{|\Gamma(G)|}e^{-\beta(E(G')-E(G))} \right\},\;\;\text{ for unlabeled graphs}
\end{align}
The acceptance probability for unlabeled graphs comes from applying detailed balance for the system at equilibrium at inverse temperature $\beta$. Let $G_u$ and $G_u'$ be two unlabeled graphs that differ by a single edge flip,
\begin{align}
\label{eq:detailedBalance}
    p(G_u)P(G_u\rightarrow G'_u)&=p(G'_u)P(G'_u\rightarrow G_u)\CR
    \frac{P(G_u\rightarrow G'_u)}{P(G'_u\rightarrow G_u)}&=\frac{p(G'_u)}{p(G_u)}=\frac{|\Gamma(G')|}{|\Gamma(G)|}e^{-\beta (E(G')-E(G))},
\end{align}
where $G',G$ are any two labeled graphs in the isomorphism classes of $G'_u$ and $G_u$ and the probability to find the system in the unlabeled graph state $G'_u$ can be written as $p(G'_u) = \frac{1}{Z^l_N}\frac{|\Gamma(G')|}{N!}e^{-\beta E(G')}$. Then, writing $P(G_u\rightarrow G'_u)=g(G_u\rightarrow G'_u)A(G_u\rightarrow G'_u)$ as the product of selection probability $g(G_u\rightarrow G'_u)$ and acceptance probability $A(G_u\rightarrow G'_u)$ and setting $g(G_u\rightarrow G'_u)$ to 
\begin{align}
    g(G_u\rightarrow G'_u)&=\begin{cases}
        0&\text{ if }|G'_u|\neq |G_u|\pm1\\
        \mathrm{const}&\text{ otherwise } 
    \end{cases}
\end{align}
combined with the Metropolis choice, gives the acceptance ratio in Eq.(\ref{eq:sampling}).  

The acceptance ratio in Eqs.(\ref{eq:sampling})–(\ref{eq:detailedBalance}), which accounts for the relative sizes of automorphism groups, is the standard Metropolis choice for sampling \emph{modulo symmetries} and coincides with the “orbital MCMC” framework \cite{DyerJerrumMueller2014}. Earlier work by Jerrum formalized uniform sampling over $G$‑orbits via Markov chains \cite{Jerrum1993LFCS}. In our setting, the group is $S_N$ acting on labeled graphs, and edge‑flip proposals move between neighboring orbits (unlabeled graphs) while preserving detailed balance with respect to the unlabeled Boltzmann weight. A validation of our implementation of sampling method is shown in Fig.\ref{fig:UnlabeledSamplingStrategyValidation} of Appendix A, 
 where for $N = 7,8,9,10$ we plot both the exact distributions and the Monte Carlo simulation results using acceptance probabilities in Eq.\ref{eq:sampling}. The exact computations are performed by evaluating the partition functions $Z^l_N, Z^u_N$ exactly by summing over all labeled (unlabeled) graphs. 

\subsection{The Free Theory}
For quantum graphs, we will consider a rescaled free Hamiltonian 
\begin{align}
    H_0=J\left(E_0n_0+E_1n_1\right) = J\left(E_0n_0+E_1(\binom{N}{2}-n_0)\right)
\end{align}
where $n_0, n_1$ are the number of edges in $G_0, G_1$ respectively, and $J$ is an $N$ dependent factor needed to make the internal energy of the system extensive in the number of vertices. Since the sum $n_0+n_1=\binom{N}{2}$ is the number of edges in the complete graph $K_N$, setting $J=2/(N-1)$ will make the internal energy extensive in $N$. Assuming that the one-particle energy $E_0$ is strictly less than $E_1$, the ground state is given by $\ket{K_N,\{\}}$, where the $\ket{0}$ one-particle state has all edges and the $\ket{1}$ one-particle state has none. The ground state is unique. At finite temperature, the labeled partition function $Z^l_N$ can be computed analytically. 
\begin{align}
    Z^l_N&=e^{-\beta J E_1 N(N-1)/2}\sum_{n_0=0}^{N(N-1)/2}\binom{N(N-1)/2}{n_0}e^{-\beta J(E_0-E_1)n_0}\CR
    &=e^{-\beta J E_1 N(N-1)/2}\left(1+e^{\beta J (E_1-E_0)}\right)^{N(N-1)/2}
\end{align}
This is simply showing that the partition function $Z^N_l$ is the partition function for a single edge state raised to the power of the total number of possible edges, $\binom{N}{2}$.
The free energy per vertex $f$, internal energy per vertex $u$, and specific heat per vertex $c$ are analytic in $\beta$ 

\begin{align}
    f&=\frac{F}{N}=J\left(\frac{N-1}{2}\right)E_1-\frac{1}{\beta}\left(\frac{2}{N-1}\right)\ln\left(1+e^{\beta J \Delta E}\right)\\
    u&=\frac{U}{N}=\frac{\partial}{\partial\beta}\left(\beta f\right)=J\left(\frac{N-1}{2}\right)\left(E_1-\Delta E\frac{1}{1+e^{-\beta J\Delta E}}\right)\\
    c&=\frac{C}{N}=-\beta^2\frac{\partial^2}{\partial\beta^2}\left(\beta f\right)=\left(\frac{N-1}{2}\right)\left(\frac{\beta J\Delta E}{2}\right)^2\mathrm{sech}^2\left(\frac{\beta J\Delta E}{2}\right),
\end{align}
where $\Delta E = E_1-E_0$. For $J=2/(N-1)$, in the thermodynamic limit $N\rightarrow\infty$, $u\rightarrow \Delta E/2$, $c\rightarrow 0$.

Let $g$ be a labeled graph with edge set $E_g$, the probability of finding this graph as a subgraph of $G_0$ is given by the expectation value of the indicator operator $\langle\m{I}^0_{g}\rangle$. 
\begin{align}
    \langle \m{I}^0_g\rangle&=\frac{1}{Z^N_l}\mathrm{Tr}(\rho \m{I}^0_g)=\sum_{G_l}\bra{G_l,G_l^c}\m{I}^0_{g}\ket{G_l,G_l^c}\frac{e^{-\beta E(G_l)}}{Z^N_l}\CR
    &=\sum_{g'\in \binom{[N]}{2}\backslash E_g}\bra{E_g\cup g',(E_g\cup g')^c}\m{I}^0_{g}\ket{E_g\cup g',(E_g\cup g')^c}\frac{e^{-\beta E(E_g\cup g',(E_g\cup g')^c)}}{Z^N_l}\CR
    &=\frac{1}{Z^N_l}\sum_{m=|E_g|}^{\binom{N}{2}}\binom{\binom{N}{2}-|E_g|}{m}e^{-\beta J\left[(|E_g|+m)E_0+(\binom{N}{2}-|E_g|-m)E_1\right]}\CR
    &=\left(\frac{1}{e^{-\beta \Delta E}+1}\right)^{|E_g|}
\end{align}
In the third line we have done the sum only over graphs where the edge set $E_g$ is in the 0 state, in which case the edge set of such a graph can be written as the disjoint union $E_g\cup g'$. The result has the simple interpretation that the probability of finding the graph $g$ in the 0 state is simply the probability that an edge is in the 0 state computed from the one-particle partition function, raised to the power of the number of edges in the graph. So, all graphs with equal number of edges have the same probability. This model is exactly the Erd\H{o}s--R\'enyi--Gilbert random graph model $G(N,p)$ \cite{Erdos:1959:pmd, 10.1214/aoms/1177706098, Fienberg01102012}, with the probability $p$ to turn on an edge given by 
\begin{align}
\label{eq:pOfBeta}
    p&=\frac{1}{1+e^{-\beta J \Delta E}},\quad \Delta E = E_1-E_0.
\end{align}

The $G(N,p)$ model has been extensively studied in the literature \cite{bollobas2013phase, bollobas2001random}. Generally, the thermodynamic limit $N\rightarrow \infty$ is taken with $p$ a function of $N$. The model exhibits threshold phenomena and phase transitions. There are {\em structural} phase transitions at $p=1/N$ and $p=\ln N/N$ that correspond to the onset of a giant connected component and the graph being fully connected, respectively. The ``order parameter" for these transitions can be thought of as the fraction of vertices in the largest connected component. In our case, since $p\in [1/2,1)$ and $p^c=1-p \in (0,1/2]$, the phase transitions signaling the onset of a giant connected component and the graph being connected will be seen in the graph $G_1$, in the $\ket{1}$ one-particle state. Let $S_1$ be the largest connected graph component of $G_1$, we define $s_1 = |S_1|/N$ to be the fraction of vertices in the largest connected component of $G_1$. $s_1$ is the ``order parameter" of the phase transition. However, note that this structural phase transition is not a proper thermodynamic phase transition. This system has no thermodynamic phase transitions since the free energy is analytic in $\beta$ and has simple thermodynamic limits at any finite $\beta$. There is also no discontinuity in the susceptibility $\chi_{s_1} \equiv \beta N \left(\langle s_1^2\rangle-\langle s_1\rangle^2\right)$ of the order parameter $s_1$ as seen in Fig.\ref{fig:FreeSimsN10To24}.

The situation is more interesting for the unlabeled quantum graph system. The partition function of the unlabeled quantum graphs for the free Hamiltonian reduces to 
\begin{align}
    Z^u_N&=e^{-\beta J E_1 N(N-1)/2}\sum_{m=0}^{N(N-1)/2}D(N,m) e^{\beta J m (E_1-E_0)}
\end{align}
where $D(N,m)$ is the number of unlabeled graphs with $N$ vertices and $m$ edges, and can be computed from the set of labeled graphs $\{g_l^m\}$ with $N$ vertices and $m$ edges as 
\begin{align}
    D(N,m)=\frac{1}{N!}\sum_{g^m_l}|\Gamma(g_l^m)|.
\end{align}
$D(N,m)$ is more efficiently computed from P\'olya's enumeration theorem,
\begin{align}
    \sum_{m=0}^{N(N-1)/2} D(N,m)x^m=Z_{S^{(2)}_N}(1+x),
\end{align}
where $Z_{S^{(2)}_N}$ is the cycle index of the action of the pair group $S^{(2)}_N$ on the edges of the complete graph $K_N$ \cite{harary1969graph, harary2014graphical}. Therefore, the unlabeled partition function is 
\begin{align}
    Z^u_N&=e^{-\beta J E_1 N(N-1)/2}Z_{S_N^{(2)}}(1+e^{\beta J\Delta E}).
\end{align}
\noindent This follows directly from the P\'olya–Redfield enumeration for the action of the pair group $S_N^{(2)}$ on the $\binom{N}{2}$ edge positions, i.e., from the cycle index $Z_{S_N^{(2)}}$ specialized at $1+e^{\beta J\Delta E}$.

We performed Monte Carlo (MC) simulations of the free labeled and unlabeled graphs at several different temperatures for $N\in\{10, 12, 15, 18, 21, 24\}$. At each temperature, the simulations involved an equilibration run followed by a run to determine the autocorrelation time $\tau$. Subsequently, $1000$ independent measurements were made (i.e., one measurement every $\tau$ MC sweep, where a sweep is $N(N-1)/2$ MC steps). Figure \ref{fig:FreeSimsN10To24} shows graphs of the internal energy per vertex $u$, specific heat per vertex $c$, 
fraction of vertices in the largest connected component of $G_1$, $s_1$ and the susceptibility of $s_1$, $\chi_{S_1}$. Both specific heat and susceptibility begin to diverge as $N$ increases, indicating a proper thermodynamic phase transition. There is also a critical slowing of the correlation pseudo-time, measured as the number of MC sweeps, as shown in \ref{fig:plotCorrTHFree}. Some studies of edge correlations in unlabeled random graphs also support the qualitative difference between labeled and unlabeled random graphs \cite{wu2023testing}. 
Recent rigorous and numerical work on unlabeled random networks has shown a \emph{first‑order} phase transition in the canonical ensemble with a prescribed average number of links, along with \emph{ensemble inequivalence} between canonical and microcanonical descriptions below a critical point \cite{EvninKrioukov2025PRL, EvninKrioukov2024arXiv}. In particular, the canonical ensemble exhibits a mixture of phases and lacks the standard percolation transition familiar from labeled Erd\H{o}s--R\'enyi--Gilbert graphs. The qualitative differences we observe here between labeled and unlabeled thermodynamics are consistent with that picture.
\begin{figure}[h!]
\centering
\begin{tabular}{cc}
    \includegraphics[width=0.48\textwidth]{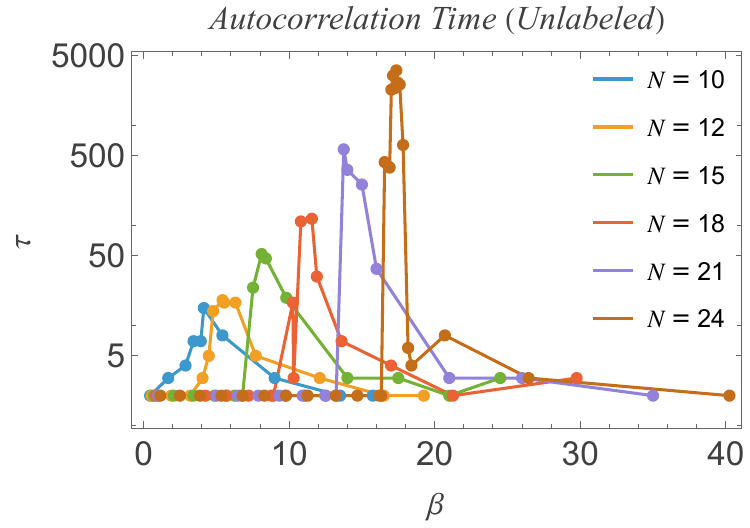}&\includegraphics[width=0.48\textwidth]{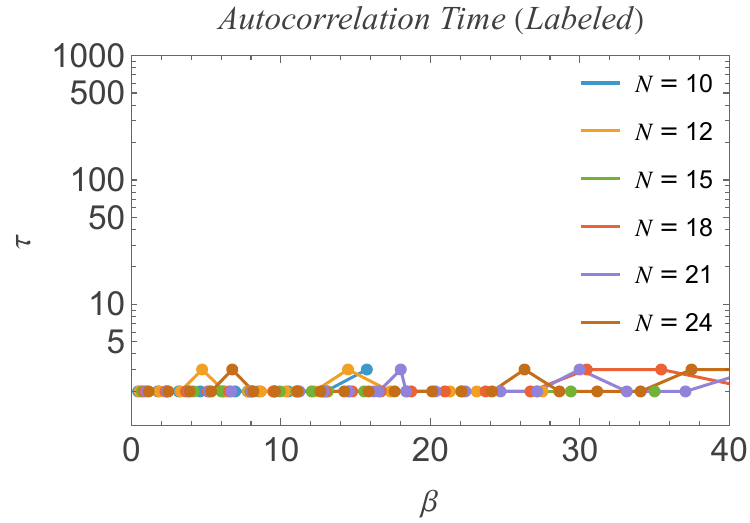}    
\end{tabular}
\caption{Graphs of the autocorrelation time $\tau$ as a function of inverse temperature $\beta$ for the free unlabeled (left) and labeled (right) graph systems.}
\label{fig:plotCorrTHFree}
\end{figure}
\begin{figure}[h!]
\centering
\begin{tabular}{cc}
    \includegraphics[width=0.45\textwidth]{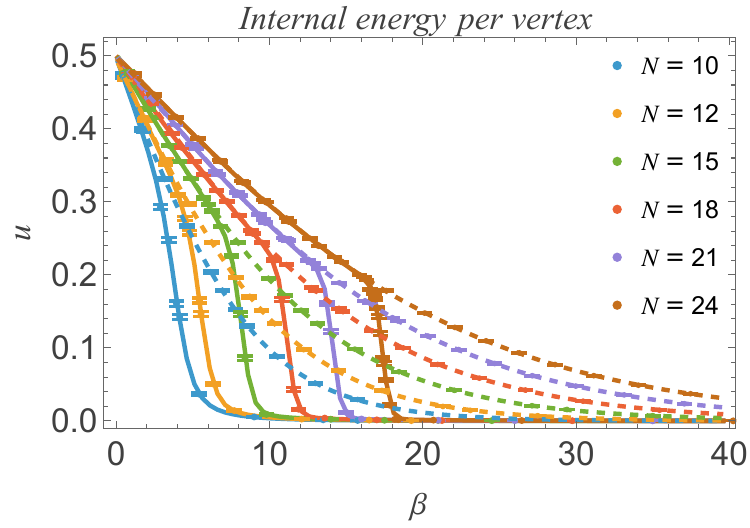}&\includegraphics[width=0.45\textwidth]{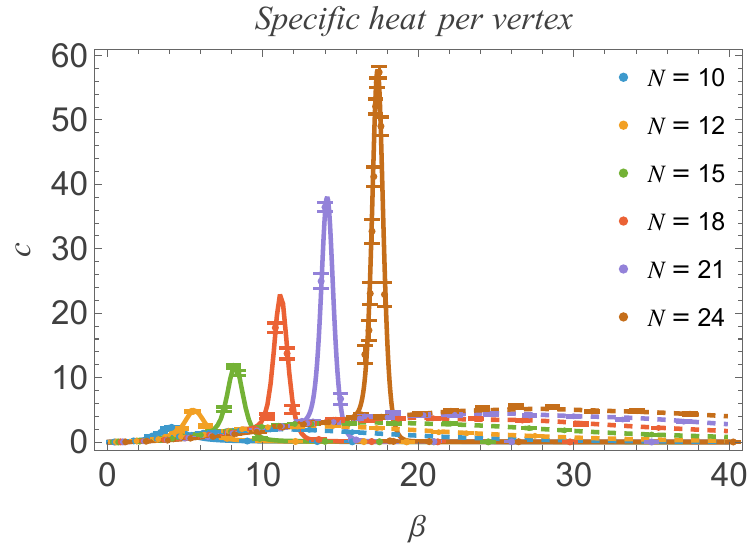}\\   
    \includegraphics[width=0.45\textwidth]{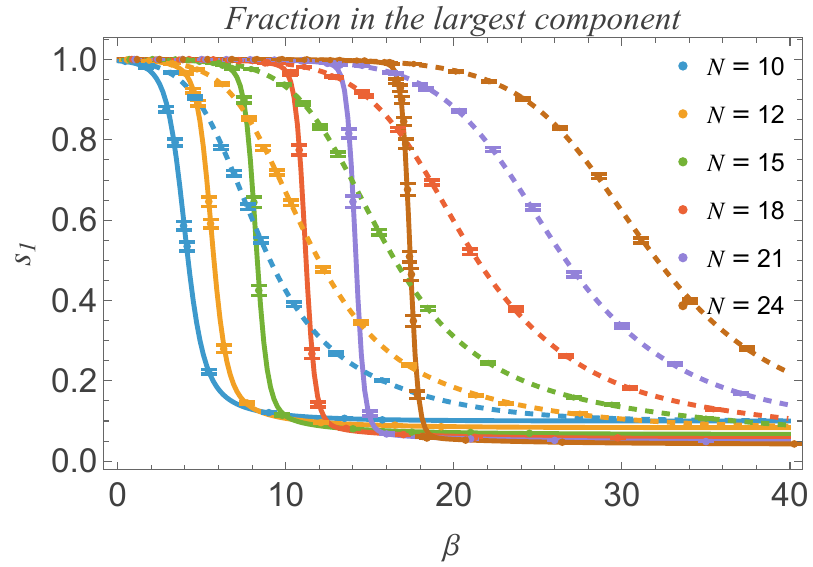}&\includegraphics[width=0.46\textwidth]{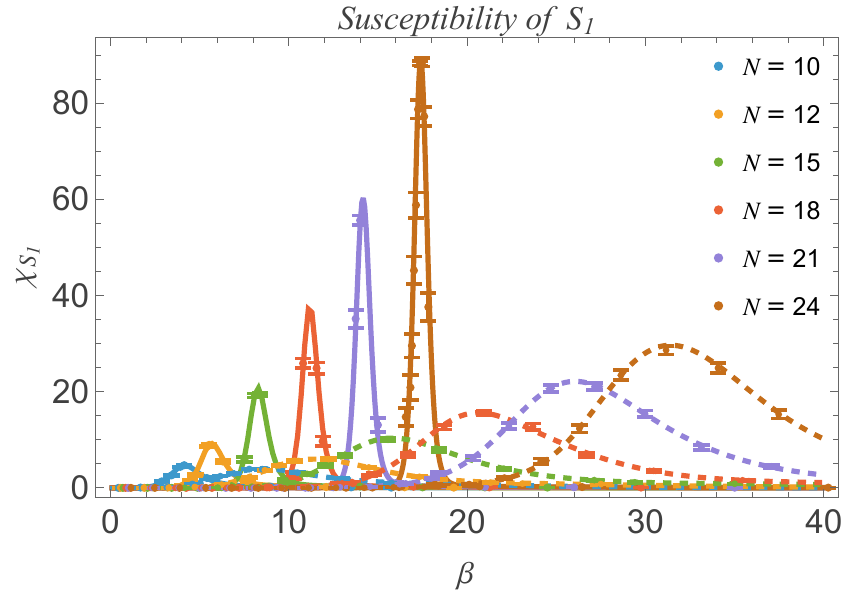}   
\end{tabular}
\caption{Graphs of $u$, $c$, $s_1$, and $\chi_{S_1}$ for the free Hamiltonian with $E_0=0,E_1=1, J=2/(N-1)$, and $N\in\{10,12,15,18,21,24\}$. Dots with error bars are Monte Carlo simulations and the interpolating lines (solid for unlabeled and dashed for labeled graphs) are extrapolations using the multiple histograms method \cite{10.1093/oso/9780198517962.001.0001}.}
\label{fig:FreeSimsN10To24}
\end{figure}

At first, such a qualitative difference between labeled and unlabeled graphs may be puzzling because in the thermodynamic limit where $N\rightarrow\infty$, the size of the automorphism group of a random graph goes to one with a probability approaching 1 \cite{erdos1963asymmetric, bollobas2001random}. Therefore, for most graphs $D(N,m)\sim \frac{1}{N!}\sum_{g_l^m}1=\frac{1}{N!}\binom{N(N-1)/2}{m}$.
However, this statement does not necessarily hold on a subclass of graphs. For instance, for $c>1$ a constant, $c N \log N \le m \le N(N-1)/2-c N \log N$, the automorphism group is likely to be non-trivial since either the graph or its complement will have isolated vertices with high probability, introducing non-trivial automorphisms \cite[Chap.~9]{bollobas1985randomgraphs}. In our case, the labeled and unlabeled distributions agree at both very small and very large $\beta$. 
However, as $\beta$ increases from small values, we start to notice significant differences in the thermodynamic properties of the labeled and unlabeled graphs. This sharp transition is related to the graph $G_1$ becoming fully connected and, hence, developing a trivial automorphism group. This is reflected in the graph of $s_1$.

\subsection{The Graph Ising Model}
We consider a rescaled version of the Ising Hamiltonian given in Eq.(\ref{eq:IsingHamiltonian}), which for quantum graphs simplifies to 
\begin{align}
\label{eq:HIsing}
    H_{Ising}&=J\left(E_0b^0+E_1b^1\right),
\end{align}
where $J$ is an $N$ dependent positive parameter needed to make the internal energy extensive in the number of vertices. We will recast the Ising Hamiltonian in two ways to highlight different limits.
\begin{align}
    H_{Ising}&=J\left(E_0+E_1\right)b^0-2J(N-2)E_1n_0+JN\binom{N-1}{2}E_1
    \CR&=J\left[\frac{E_0}{2}\sum_i (d^0_i)^2+\frac{E_1}{2}\sum_i(d^1_i)^2\right]-J\left(E_0n_0+E_1n_1\right)
\end{align}
where $n_0, n_1$ are the number of edges in $G_0,G_1$ and $d^0_i, d^1_i$ are the degree of vertex $i$ in $G_0, G_1$ respectively. In the special case where $E_0=-E_1$ we recover the free system with $E_0=0$ up to a constant shift in energy. If $E_0 < 0$, regardless of the value of $E_1\ge E_0$ the system is ``ferromagnetic" in the sense that there is a unique ground state $\ket{K_N,\{\}}$ (possibly including its mirror $\ket{\{\},K_N}$ if $E_0=E_1$ ). This can be seen more easily in Eq.\ref{eq:HIsing}, where maximizing the number of angles $b^0$ in $G_0$ minimizes the energy.  For $0=E_0<E_1$, the ground state is not unique. There are $\lfloor N/2\rfloor+1$ states with the minimum energy corresponding to all angle-free graphs (i.e., all graphs with a maximum degree of 1) in the $\ket{1}$ state. For $0< E_0< E_1$, since $n_0+n_1=\binom{N}{2}$ is fixed, the Hamiltonian is minimized when both $G_0$ and $G_1$ are as close to a regular graph as possible, thus simultaneously minimizing $\sum_i (d^0_i)^2$ and $\sum_i (d^1_i)^2$. $G_0$ will have more edges overall depending on how large $\Delta E = E_1-E_0$ is. The larger $\Delta E$, the larger $n_0-n_1$. In the degenerate case $0<E_0=E_1$, the ground-state graphs have both edges and degrees distributed as evenly as possible between $G_0$ and the complement $G_1$. The ground state is highly degenerate, and the system exhibits frustration. Fig.\ref{fig:IsingGroundStates} shows the ground states for $N=7, E_1 = 1$, and $E_0\in\{-1, 0,1/2, 1\}$.

\begin{figure}
\centering
    \includegraphics[width=0.95\textwidth]{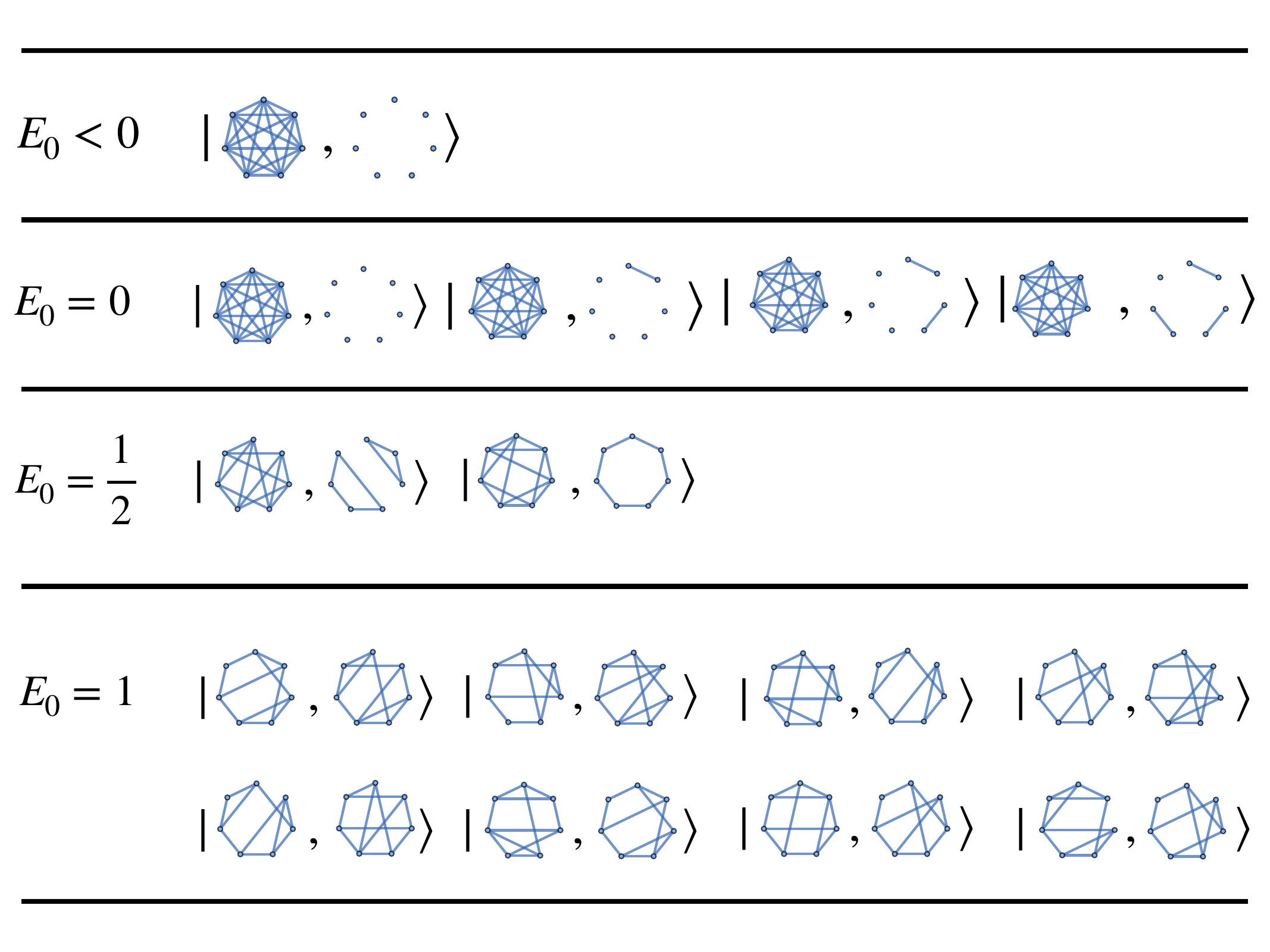} 
\caption{Ground states of the graph Ising model for $N=7, E_1=1$, and $E_0 \in \{-1, 0,1/2, 1\}$. There is a unique ground state for $E_0<0$, but the ground state is degenerate for $E_0\ge 0$.}
\label{fig:IsingGroundStates}
\end{figure}

An equivalent way to view this model is as the traditional Ising model of spins sitting on the vertices of the line graph of the complete graph, as will be shown below. The line graph of a graph has the edges of the original graph as its vertices, and two vertices are connected in the line graph if their corresponding edges in the original graph share a vertex. The complete graph $K_4$ and its line graph $L(K_4)$ are shown in Fig.\ref{fig:KNLineGraph}. Each vertex in $L(K_N)$ has degree $2(N-2)$. So for large $N$, the system is well approximated by the infinite-dimensional Ising model, which has analytic solutions in both the ferromagnetic and antiferromagnetic cases \cite{infinite_range_ising_ubc,huang1987statistical,PhysRevE.65.027102,PhysRevE.63.026116}. Let $L_{ij}$ be the adjacency matrix element of the line graph of the complete graph, and let $\m{I}^k_i$ be the indicator operator as before where the index $i$ now labels the edges. Then, the Ising Hamiltonian can be recast as
\begin{align}
\label{eq:HIsingQuantumGraph}
    H_{Ising}&=\frac{J}{2}\left(E_0\sum_{ij}\m{I}^0_iL_{ij}\m{I}^0_j+E_1\sum_{ij}\m{I}^1_iL_{ij}\m{I}^1_j\right)\CR
    &=\frac{J}{2}(E_0+E_1)\sum_{i,j}\m{I}^0_iL_{ij}\m{I}^0_j-2JE_1(N-2)n^0+JN\binom{N-1}{2}E_1,
\end{align}
where we used $\m{I}^0_i+\m{I}^1_i=1$ on the second line. Ignoring the last term which is constant shift and introducing new ``spin" variables $S_i=2\m{I}^0_i-1$, so that $S_i\in\{-1,+1\}$, we get
\begin{align}
\label{eq:HIsingLineGraphRecast1}
     H_{Ising}&=\frac{J}{8}(E_0+E_1)\sum_{ij}S_iL_{ij}S_j-J\frac{N-2}{2}(E_0+3E_1)\sum_j S_j
\end{align}
up to an overall constant shift. This has a similar form to an infinite range Ising model with a Hamiltonian of the form 
\begin{align}
\label{eq:HIsingLineGraphRecast2}
    H=\sum_{ij}S_i J_{ij}S_j-h\sum_{j}S_j
\end{align}
where for the infinite-range Ising model $J_{ij}=JK_{ij}$, with $K_{ij}$ the adjacency matrix of the complete graph. In our case, $J_{ij}$ is the rescaled adjacency matrix of the line graph of the complete graph, and $h$ is a fixed external magnetic field,
\begin{align}
    J_{ij}&=\frac{J}{8}(E_0+E_1)L_{ij},\quad h=J\frac{N-2}{2}(E_0+3E_1).
\end{align}

As Eqs. (\ref{eq:HIsingLineGraphRecast1})-(\ref{eq:HIsingLineGraphRecast2}) show, the interactions live on $L(K_N)$, i.e., the Ising model on the line graph of a complete graph. For antiferromagnetic couplings on arbitrary line graphs, polynomial‑time approximation schemes for the partition function and rapid mixing of Glauber dynamics are known \cite{DyerHeinrichJerrumMueller2021}. In the ferromagnetic direction, the high degree $2(N-2)$ of $L(K_N)$ makes a mean‑field (Curie–Weiss) analysis accurate in the large‑$N$ limit; (see standard expositions of the Curie–Weiss model for the emergence of the self‑consistency equation and thermodynamic singularities \cite{VelenikCurieWeiss}).

\begin{figure}
\centering
\includegraphics[width=0.8\textwidth]{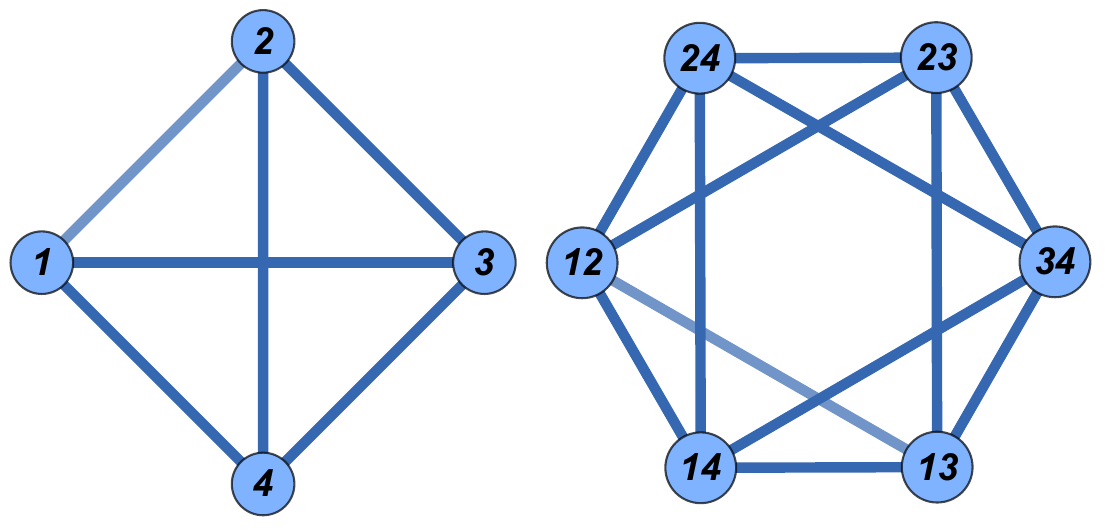} 
\caption{The complete graph $K_4$ (left) and its line graph $L(K_4)$ (right).}
\label{fig:KNLineGraph}
\end{figure}
The infinite-range Ising ferromagnet and antiferromagnet have been studied in \cite{PhysRevE.63.026116, PhysRevE.65.027102}. These systems have Ising spins sitting on the vertices of the complete graph and interact all-to-all. The vertices are treated as indistinguishable. The ferromagnetic case has a second-order phase transition. The antiferromagnetic system does not have a phase transition. 

Numerical simulations show that these qualitative properties of the infinite-range Ising model also hold for the unlabeled ferromagnetic Ising graph model in our case. On the other hand, the labeled Ising graphs, both ferromagnetic and antiferromagnetic, exhibit no phase transitions. We perform Monte Carlo numerical simulations for $N\in\{10,12,15,18,21,24\}$ for labeled and unlabeled quantum graphs with $E_1=1.0$ fixed and $E_0$ varied between $-0.5$ and $0.5$ to study the ferromagnetic and antiferromagnetic cases, respectively. We also set $J=1/\binom{N-1}{2}$ to make the internal energy extensive in $N$. The reason is that, as can be seen from the first line of Eq.\ref{eq:HIsing}, the internal energy $U$ and the specific heat $C$ of this system grow with the number of angle subgraphs of the complete graph. There are $N\binom{N-1}{2}$ angle subgraphs in the complete graph, so setting $J=1/\binom{N-1}{2}$ will make $U$ and $C$ proportional to $N$.  
In addition to the internal energy per vertex $u$, specific heat per vertex $c$, fraction of vertices in the largest connected component $s_1$, and its susceptibility $\chi_{S_1}$, we also plot the average magnetization per edge $m$ of graph $G_1$, where $m=n_1/\binom{N}{2}\in[0,1]$.  

\subsubsection{Ferromagnetic system}

 Similarly to the infinite-range Ising model, the unlabeled graph ferromagnetic system has a sharp transition marked by critical slowing, divergence of the specific heat per vertex $c$, and divergence of susceptibilities $\chi_{S_1}$ and $\chi_m\equiv \beta\binom{N}{2}\left(\langle m^2\rangle-\langle m\rangle^2\right)$ (see Fig.\ref{fig:IsingSimFerroPlots}). This is strongly suggestive of an actual second-order phase transition in the thermodynamic limit of $N\rightarrow \infty$, though, it is not clear whether the critical temperature $\beta_c$ goes to infinity or converges to a finite value in that limit. However, this is not the case for the labeled graphs. The ferromagnetic labeled graphs do not appear to have any thermodynamic phase transitions. There is no critical slowing, and the specific heat and susceptibility behave smoothly with increasing $N$.

\begin{figure}[h]
\centering
\begin{tabular}{cc}
    \includegraphics[width=0.45\textwidth]{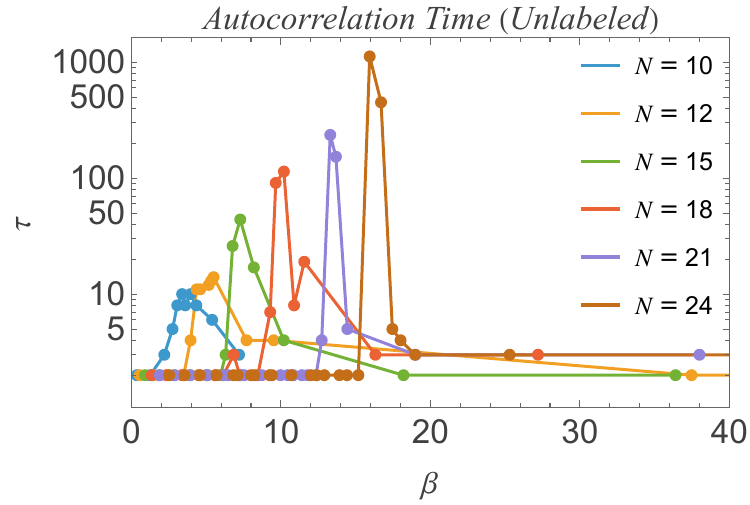}&\includegraphics[width=0.45\textwidth]{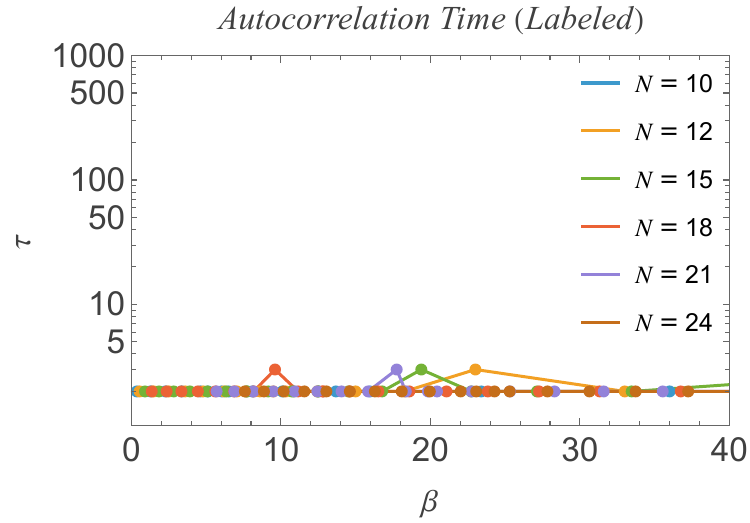}
\end{tabular}
\caption{Plots of autocorrelation times for the unlabeled (left) and labeled (right) ferromagnetic Ising graph systems.}
\label{fig:IsingFerroCorrT}
\end{figure}

\begin{figure}[h]
\centering
\begin{tabular}{cc}
    \includegraphics[width=0.45\textwidth]{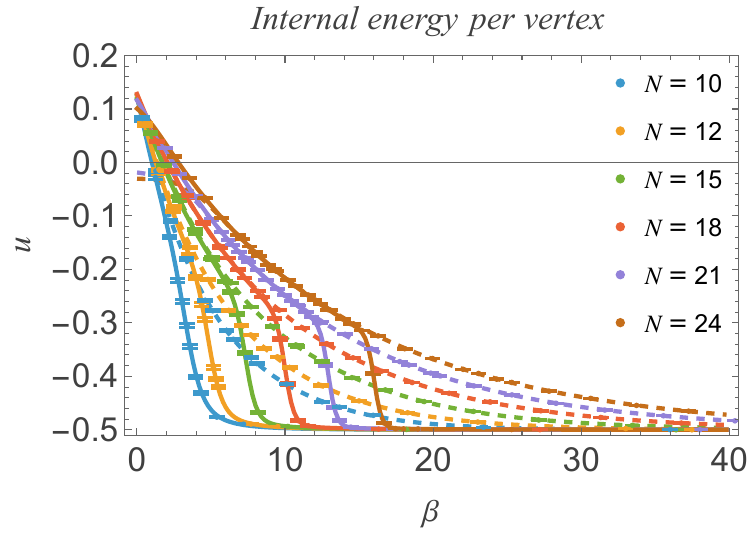}&\includegraphics[width=0.45\textwidth]{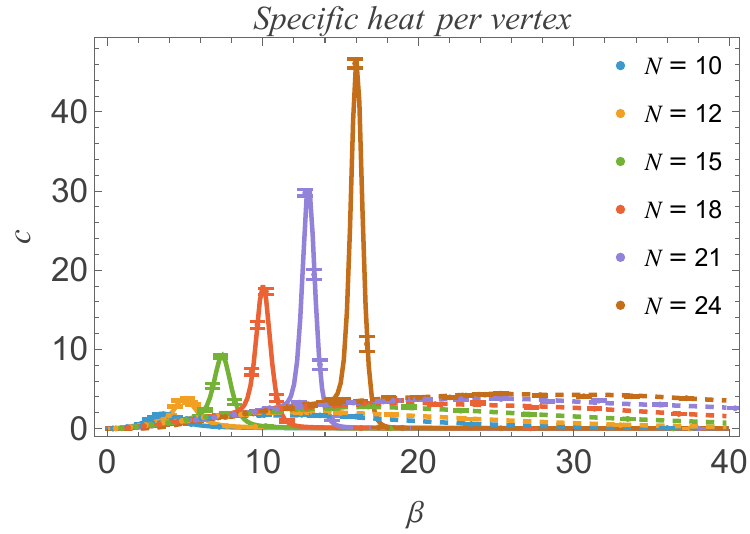}\\ 
    \includegraphics[width=0.45\textwidth]{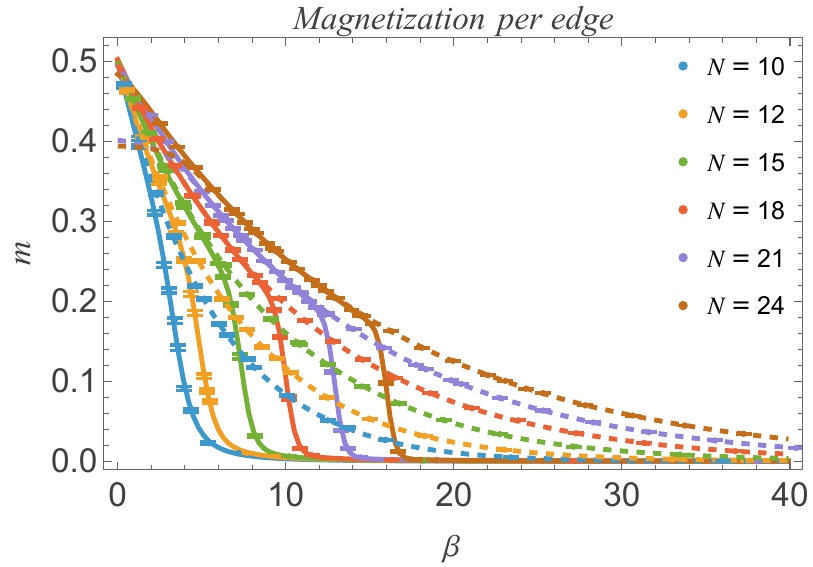}&\includegraphics[width=0.45\textwidth]{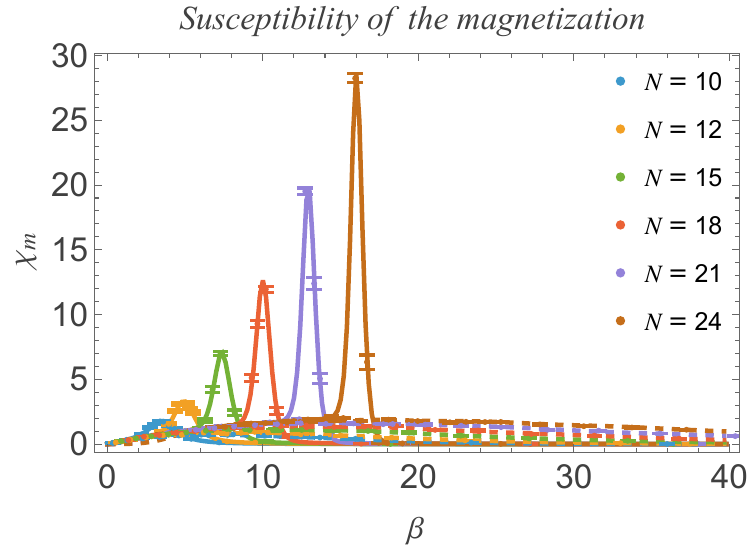}\\
    \includegraphics[width=0.45\textwidth]{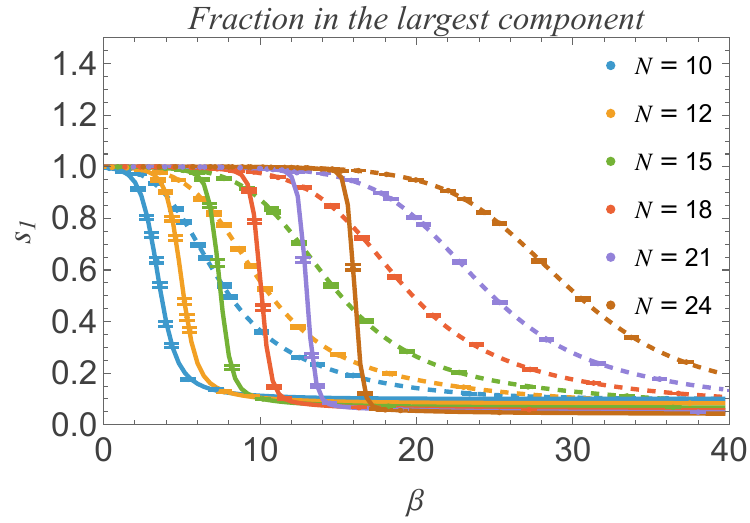}&\includegraphics[width=0.45\textwidth]{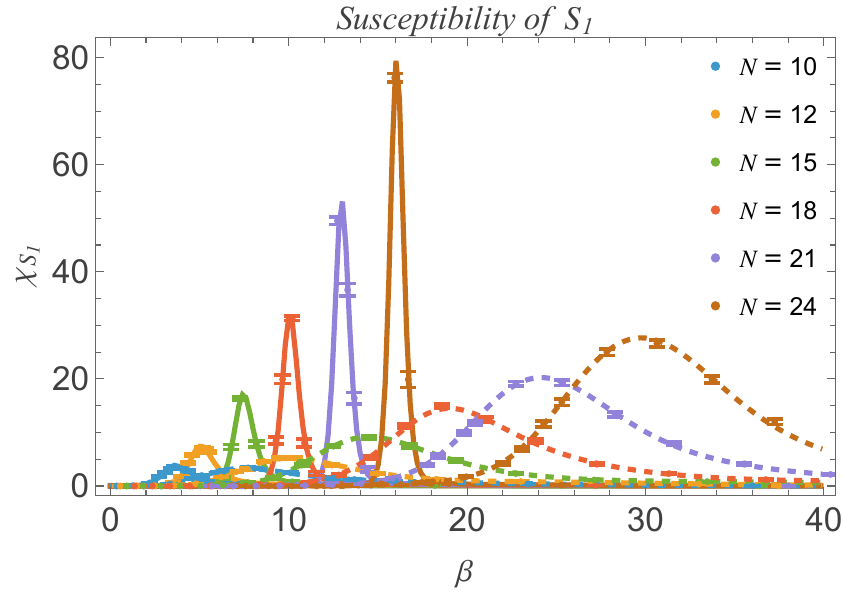}
\end{tabular}
\caption{(Left column) Graphs of the internal energy per vertex $(u)$, magnetization of $G_1$, and fraction of vertices in the largest component of $G_1$, $(s_1)$; (right column) specific heat per vertex ($c$), susceptibility of the magnetization ($\chi_m$), and susceptibility of the fraction of vertices in the largest component of $G_1$, ($\chi_{s_1}$) for the ferromagnetic labeled and unlabeled quantum graphs with $E_0=-0.5,E_1=1$, and $N\in\{10,12,15,18,21,24\}$. The dots with error bars are Monte Carlo simulations, the interpolating lines (solid for unlabeled and dashed for labeled graphs) are made using the multiple histograms method \cite{10.1093/oso/9780198517962.001.0001}.}
\label{fig:IsingSimFerroPlots}
\end{figure}

The close resemblance of the graphs of the internal energy and mean magnetization is not surprising since they are closely related by Eq.(\ref{eq:HIsingDegreeVsAngleCount}). 
\vspace{16em}

\subsubsection{Antiferromagnetic system}
Fig.\ref{fig:IsingAntiFerroPlots} shows the MC simulation results for unlabeled and labeled antiferromagnetic graph systems at $E_0=+0.5, E_1=1.0$. There is no evidence of a phase transition in both the unlabeled and labeled antiferromagnetic Ising graph systems. As $N$ increases, there is little difference in the graphs of the various thermodynamic variables between the labeled and unlabeled graphs. That is because the antiferromagnetic graph system has connected $G_0$ and $G_1$ in the ground states, and therefore connected graphs at all temperatures. A random connected labeled graph has a trivial automorphism group with a probability approaching 1 as $N\rightarrow \infty$. Therefore, the qualitative difference between the labeled and unlabeled graph systems seen in the ferromagnetic case does not arise here.

\begin{figure}[h!]
\centering
\begin{tabular}{cc}
    \includegraphics[width=0.45\textwidth]{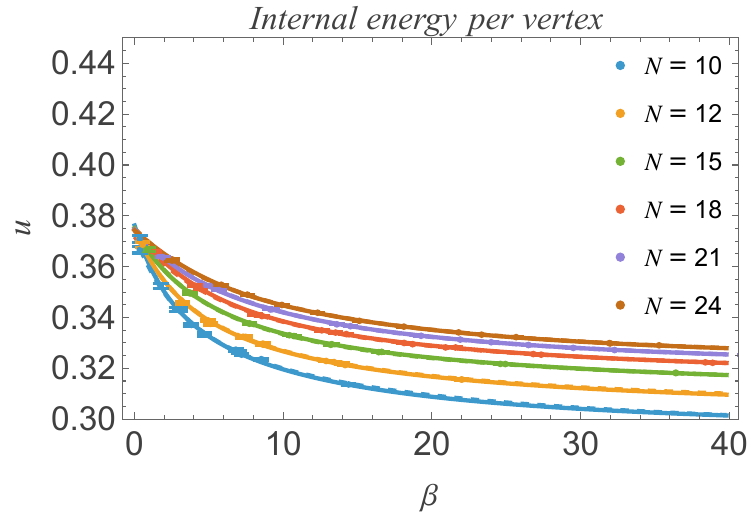}&\includegraphics[width=0.45\textwidth]{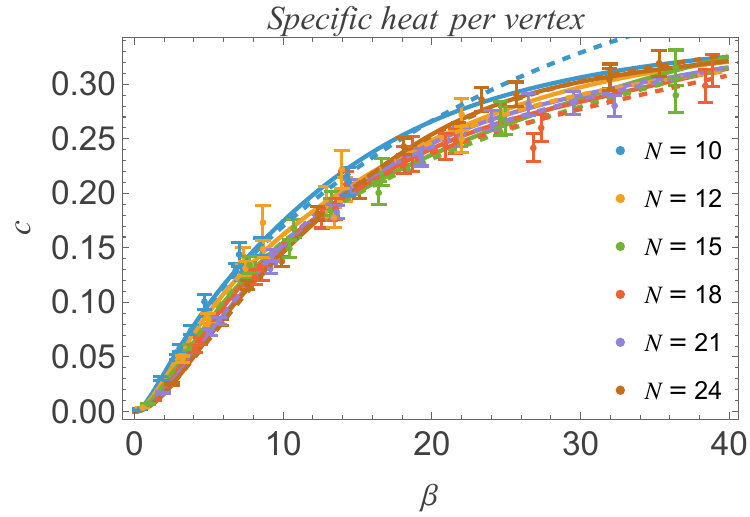}\\
    \includegraphics[width=0.45\textwidth]{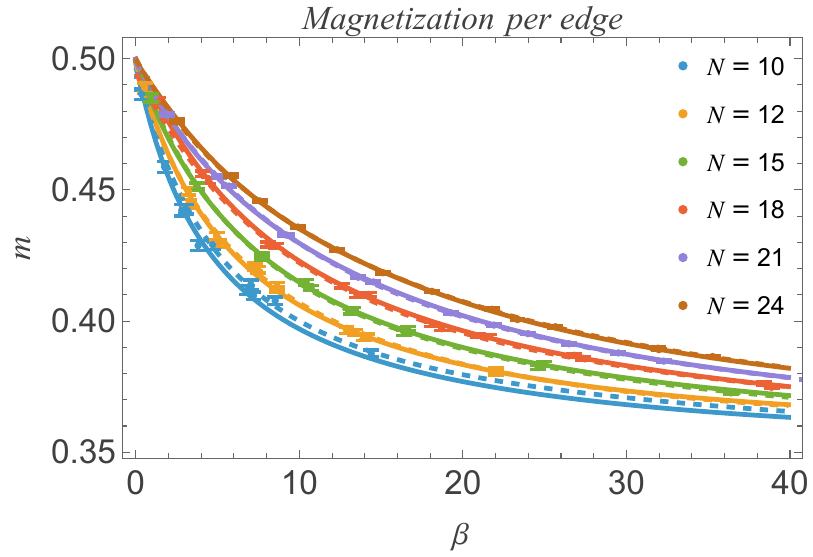}&\includegraphics[width=0.45\textwidth]{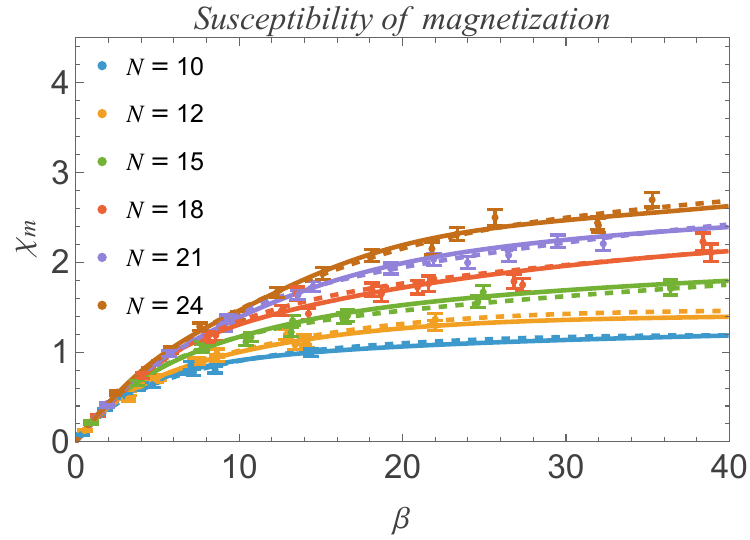}
    \\
    \includegraphics[width=0.45\textwidth]{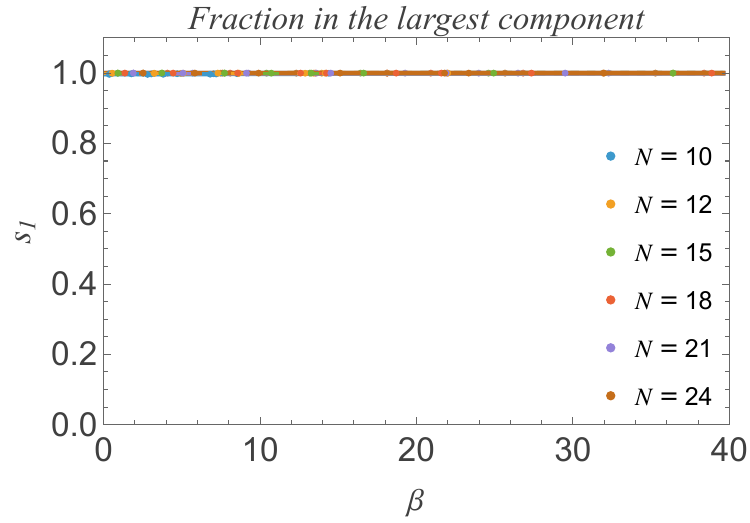}&\includegraphics[width=0.45\textwidth]{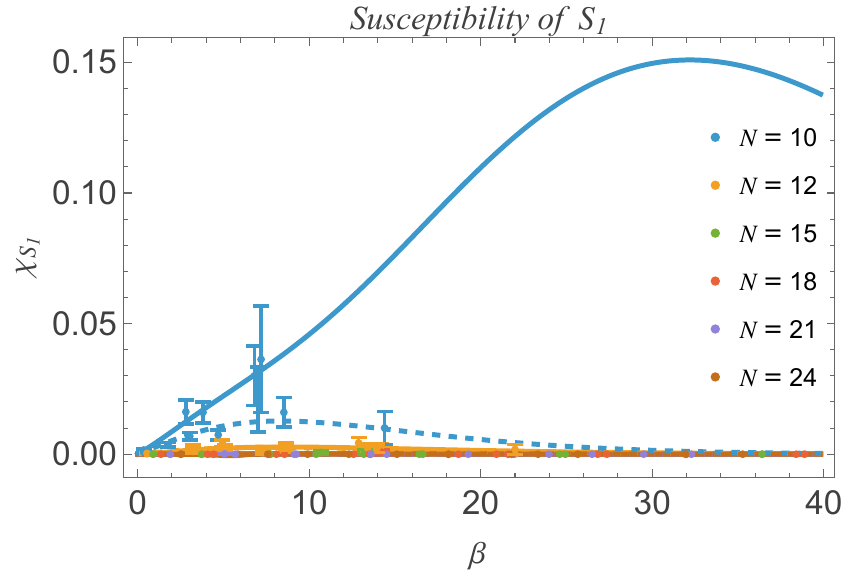}
\end{tabular}
\caption{(Left column) Graphs of the internal energy per vertex $(u)$, mean magnetization of $G_1$ ($m$), and fraction of vertices in the largest connected component ($s_1$); (Right column) Graphs of specific heat per vertex ($c$), susceptibility of $m$ ($\chi_m$), and susceptibility of $s_1$ ($\chi_{s_1}$) for antiferromagnetic Ising graph system with $E_0=+0.5,E_1=1.0$, and $N\in\{10,12,15,18,21,24\}$. The dots with error bars are Monte Carlo simulations, and interpolating lines (solid for unlabeled, dashed for labeled) are made using the multiple histograms extrapolation method. The dashed lines for larger $N$ values are not visible as they overlap with the solid lines due to the two distributions converging.}
\label{fig:IsingAntiFerroPlots}
\end{figure}

\begin{figure}[h!]
\centering
\begin{tabular}{cc}
    \includegraphics[width=0.45\textwidth]{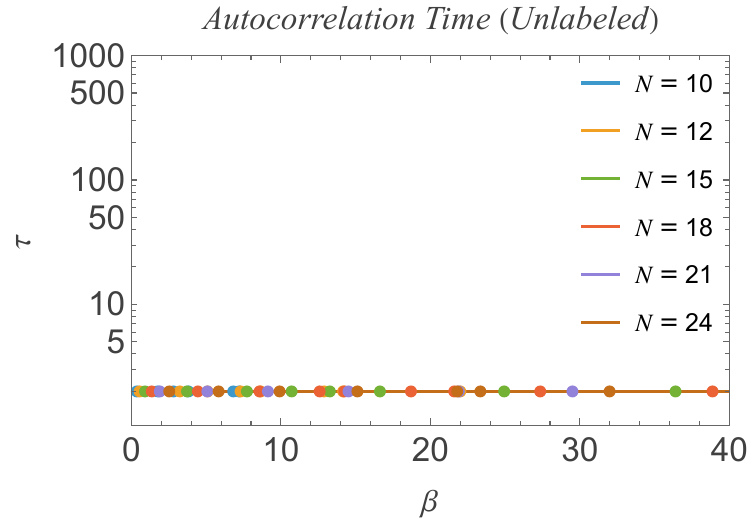}&\includegraphics[width=0.45\textwidth]{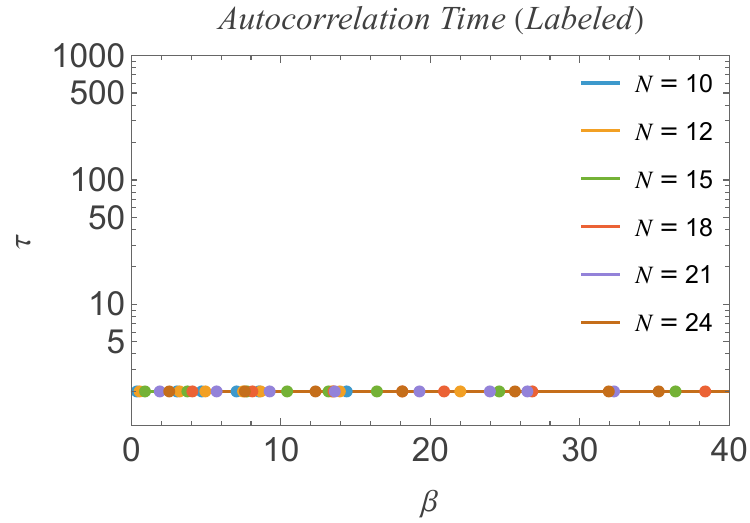}
\end{tabular}
\caption{Graphs of the autocorrelation time $\tau$ in Monte Carlo time for the unlabeled (left column) and labeled (right column) for the Ising antiferromagnetic graph system with $E_0=+0.5,E_1=1.0$, and $N\in\{10,12,15,18,21,24\}$.}
\label{fig:IsingAntiFerro}
\end{figure}

\vspace{5em}

\section{Summary and Future Work}
\label{sec:conclusion}
We extended the notion of quantum labeled graphs to labeled and unlabeled quantum multigraphs. The Hilbert space associated with the labeled finite multigraphs for a fixed number of vertices $N$ and dimension $D$ of the one-particle Hilbert space has basis kets given by all weak ordered set partitions of the edge space of the complete graph into $D$ blocks. The Hilbert space of the unlabeled quantum multigraphs is found by projecting the labeled Hilbert space to the 1D irreducible subspaces under the action of the symmetric group. The projection operators are the symmetrization and antisymmetrization operators. 

We introduced simple dynamics on the Hilbert space of quantum graphs; the free Hamiltonian which is inherited from the Hamiltonian of the one-particle edge Hilbert space with no interaction among different edges, and the graph Ising model where neighboring edges interact with a simple Ising-type Hamiltonian. The thermodynamic properties of the quantum graph systems were investigated numerically using Monte Carlo simulations. We found strong evidence of a proper thermodynamic phase transition in the unlabeled graphs for both the free Hamiltonian and the ferromagnetic Ising Hamiltonian at possibly infinite inverse temperature (at zero temperature), characterized by divergence of autocorrelation time, specific heat, and susceptibilities. This phase transition has to do with $G_1$ transitioning from having a non-trivial automorphism group to a trivial automorphism group. This is further corroborated by the absence of a phase transition in the antiferromagnetic system in which both $G_0$ and $G_1$ remain connected (and hence typically have a trivial automorphism group) at all temperature ranges from the ground state to infinite temperature. The labeled graph system showed no evidence of a phase transition for both the free and the Ising Hamiltonians. 

The similarity of the graphs of the thermodynamic functions in the free and Ising (ferromagnetic) Hamiltonian leads us to speculate that the unlabeled graphs are in the same universality class for any microscopic Hamiltonian as long as it has the unique ground state given by $\ket{K_N,\{\}}$. This initial numerical exploration can be extended in future work to explore this speculation further. First, it will be interesting to determine whether the critical temperature at the thermodynamic limit is finite or zero by performing the mean field theory computation analytically. The mean field theory approximation will also give the critical exponents and address the question of the universality classes. Another direction to pursue is to extend the numerical investigation of the thermodynamics of quantum multigraphs for $D\ge2$. Further down the line, we are interested in using quantum multigraphs to construct models of emergent geometry where the goal there will be to identify a class of Hamiltonians whose ground-state graphs are combinatorial manifolds.  

\section*{Acknowledgement}
This work is supported in part by the National Science Foundation LEAPS-MPS grant (award number 2138323) and the DOE RENEW-HEP grant (award number DE-SC0024518). KB thanks SJSU colleagues Curtis Asplund, Alejandro Garcia, Ehsan Khatami, Chris Smallwood, Ken Wharton; and Pepperdine Colleague Kevin Iga for valuable feedback on earlier drafts of the article.

\bibliography{./bibliography}

\begin{thebibliography}{10}

\bibitem{Baez1999SpinFoams}
J.~C. Baez.
\newblock An introduction to spin foam models of quantum gravity and {BF}
  theory.
\newblock In {\em Geometry and Quantum Physics}, volume 543 of {\em Lecture
  Notes in Physics}, pages 25--93. Springer, Berlin, 2000.

\bibitem{Bichon2003PAMS}
Julien Bichon.
\newblock Quantum automorphism groups of finite graphs.
\newblock {\em Proc. Am. Math. Soc.}, 131(3):665--673, 2003.

\bibitem{bollobas1985randomgraphs}
B{\'e}la Bollob{\'a}s.
\newblock {\em Random Graphs}.
\newblock Academic Press, London, 1985.

\bibitem{bollobas1998modern}
B{\'e}la Bollob{\'a}s.
\newblock {\em Modern graph theory}, volume 184.
\newblock Springer Science \& Business Media, 1998.

\bibitem{bollobas2001random}
B{\'e}la Bollob{\'a}s.
\newblock Random graphs, volume 73 of.
\newblock {\em Cambridge studies in advanced mathematics}, 2001.

\bibitem{bollobas2013phase}
B{\'e}la Bollob{\'a}s and Oliver Riordan.
\newblock The phase transition in the erdos-r{\'e}nyi random graph process.
\newblock {\em Erd{\"o}s centennial}, 25:59--110, 2013.

\bibitem{brunnemann2010oriented}
Johannes Brunnemann and David Rideout.
\newblock Oriented matroids—combinatorial structures underlying loop quantum
  gravity.
\newblock {\em Classical and Quantum Gravity}, 27(20):205008, 2010.

\bibitem{duan2012zero}
Runyao Duan, Simone Severini, and Andreas Winter.
\newblock Zero-error communication via quantum channels, noncommutative graphs,
  and a quantum lov{\'a}sz number.
\newblock {\em IEEE Transactions on Information Theory}, 59(2):1164--1174,
  2012.

\bibitem{DyerHeinrichJerrumMueller2021}
Martin Dyer, Marc Heinrich, Mark Jerrum, and Haiko M{\"u}ller.
\newblock Polynomial-time approximation algorithms for the antiferromagnetic
  ising model on line graphs.
\newblock {\em Combinatorics, Probability and Computing}, 30(6):905--921, 2021.

\bibitem{Erdos:1959:pmd}
P.~Erd\H{o}s and A.~R\'enyi.
\newblock On random graphs i.
\newblock {\em Publicationes Mathematicae Debrecen}, 6:290--297, 1959.

\bibitem{erdos1963asymmetric}
P.~Erd\H{o}s and A.~R{\'e}nyi.
\newblock Asymmetric graphs.
\newblock {\em Acta Mathematica Academiae Scientiarum Hungaricae},
  14(3-4):295--315, 1963.

\bibitem{EvninKrioukov2024arXiv}
Oleg Evnin and Dmitri Krioukov.
\newblock Ensemble inequivalence and phase transitions in unlabeled networks,
  2024.

\bibitem{EvninKrioukov2025PRL}
Oleg Evnin and Dmitri Krioukov.
\newblock Ensemble inequivalence and phase transitions in unlabeled networks.
\newblock {\em Phys. Rev. Lett.}, 134:207401, 2025.

\bibitem{feynman1998statistical}
R.P. Feynman.
\newblock {\em Statistical Mechanics: A Set Of Lectures}.
\newblock Advanced Books Classics. Avalon Publishing, 1998.

\bibitem{Fienberg01102012}
Stephen~E. Fienberg.
\newblock A brief history of statistical models for network analysis and open
  challenges.
\newblock {\em Journal of Computational and Graphical Statistics},
  21(4):825--839, 2012.

\bibitem{VelenikCurieWeiss}
Sacha Friedli and Yvan Velenik.
\newblock {\em Statistical Mechanics of Lattice Systems: A Concrete
  Mathematical Introduction}.
\newblock Cambridge University Press, Cambridge, 2018.

\bibitem{10.1214/aoms/1177706098}
E.~N. Gilbert.
\newblock {Random Graphs}.
\newblock {\em The Annals of Mathematical Statistics}, 30(4):1141 -- 1144,
  1959.

\bibitem{goswami2023quantum}
Debashish Goswami and Sk~Asfaq~Hossain.
\newblock Quantum symmetry in multigraphs (part i).
\newblock {\em arXiv e-prints}, pages arXiv--2302, 2023.

\bibitem{goswami2024quantum}
Debashish Goswami and Sk~Asfaq Hossain.
\newblock Quantum symmetry in multigraphs (part ii).
\newblock {\em Infinite Dimensional Analysis, Quantum Probability and Related
  Topics}, 27:2440012, 2024.

\bibitem{harary1969graph}
Frank Harary.
\newblock {\em Graph Theory}.
\newblock Addison-Wesley series in mathematics. Addison-Wesley Publishing
  Company, 1969.

\bibitem{harary2014graphical}
Frank Harary and Edgar~M Palmer.
\newblock {\em Graphical enumeration}.
\newblock Elsevier, 2014.

\bibitem{hartle2022simplicial}
James~B Hartle.
\newblock Simplicial quantum gravity.
\newblock {\em arXiv preprint arXiv:2201.00226}, 2022.

\bibitem{Hein2006GraphStatesReview}
M.~Hein, W.~D{\"u}r, J.~Eisert, R.~Raussendorf, M.~Van~den Nest, and H.~J.
  Briegel.
\newblock Entanglement in graph states and its applications, 2006.

\bibitem{Hein2004GraphStates}
M.~Hein, J.~Eisert, and H.~J. Briegel.
\newblock Multiparty entanglement in graph states.
\newblock {\em Phys. Rev. A}, 69:062311, 2004.

\bibitem{huang1987statistical}
Kerson Huang.
\newblock {\em Statistical Mechanics}.
\newblock John Wiley \& Sons, New York, 2nd edition, 1987.

\bibitem{Jerrum1993LFCS}
Mark Jerrum.
\newblock Uniform sampling modulo a group of symmetries using markov chain
  simulation.
\newblock Technical Report ECS-LFCS-93-272, Laboratory for Foundations of
  Computer Science, University of Edinburgh, 1993.

\bibitem{konopka2008quantum}
Tomasz Konopka, Fotini Markopoulou, and Simone Severini.
\newblock Quantum graphity: a model of emergent locality.
\newblock {\em Physical Review D}, 77(10):104029, 2008.

\bibitem{konopka2006quantum}
Tomasz Konopka, Fotini Markopoulou, and Lee Smolin.
\newblock Quantum graphity.
\newblock {\em arXiv preprint hep-th/0611197}, 2006.

\bibitem{kuchment2008quantum}
Peter Kuchment.
\newblock Quantum graphs: an introduction and a brief survey.
\newblock {\em arXiv preprint arXiv:0802.3442}, 2008.

\bibitem{lee2009emergence}
Sung-Sik Lee.
\newblock Emergence of gravity from interacting simplices.
\newblock {\em International Journal of Modern Physics A}, 24(23):4271--4286,
  2009.

\bibitem{PhysRevE.63.026116}
Edoardo Milotti.
\newblock Exactly solved dynamics for an infinite-range spin system.
\newblock {\em Phys. Rev. E}, 63:026116, Jan 2001.

\bibitem{PhysRevE.65.027102}
Edoardo Milotti.
\newblock Exactly solved dynamics for an infinite-range spin system. ii.
  antiferromagnetic interactions.
\newblock {\em Phys. Rev. E}, 65:027102, Jan 2002.

\bibitem{musto2018compositional}
Benjamin Musto, David Reutter, and Dominic Verdon.
\newblock A compositional approach to quantum functions.
\newblock {\em Journal of Mathematical Physics}, 59(8), 2018.

\bibitem{10.1093/oso/9780198517962.001.0001}
M~E~J Newman and G~T Barkema.
\newblock {\em Monte Carlo Methods in Statistical Physics}.
\newblock Oxford University Press, 02 1999.

\bibitem{DyerJerrumMueller2014}
Mathias Niepert.
\newblock Markov chains on orbits of permutation groups.
\newblock {\em arXiv preprint arXiv:1408.2052}, 2014.

\bibitem{nieto2011oriented}
Juan~A Nieto.
\newblock Oriented matroid theory and loop quantum gravity in (2+ 2) and eight
  dimensions.
\newblock {\em Revista mexicana de f{\'\i}sica}, 57(5):400--405, 2011.

\bibitem{PhysRev.65.117}
Lars Onsager.
\newblock Crystal statistics. i. a two-dimensional model with an order-disorder
  transition.
\newblock {\em Phys. Rev.}, 65:117--149, Feb 1944.

\bibitem{plesch2003entangled}
Martin Plesch and Vladim{\'\i}r Bu{\v{z}}ek.
\newblock Entangled graphs: Bipartite entanglement in multiqubit systems.
\newblock {\em Physical Review A}, 67(1):012322, 2003.

\bibitem{Qu2013PRA}
R.~Qu, J.~Wang, Z.-S. Li, and Y.-R. Bao.
\newblock Encoding hypergraphs into quantum states.
\newblock {\em Phys. Rev. A}, 87:022311, 2013.

\bibitem{Rossi2013NJP}
M.~Rossi, M.~Huber, D.~Bru{\ss}, and C.~Macchiavello.
\newblock Quantum hypergraph states.
\newblock {\em New J. Phys.}, 15:113022, 2013.

\bibitem{RovelliSmolin1995PRD}
C.~Rovelli and L.~Smolin.
\newblock Spin networks and quantum gravity.
\newblock {\em Phys. Rev. D}, 52:5743--5759, 1995.

\bibitem{infinite_range_ising_ubc}
Gordon Semenoff.
\newblock The infinite range ising model, landau theory and critical exponents.
\newblock Lecture Notes.
\newblock Accessed: \today.

\bibitem{streater2000pct}
R.F. Streater and A.S. Wightman.
\newblock {\em PCT, Spin and Statistics, and All that}.
\newblock Princeton landmarks in mathematics and physics. Princeton University
  Press, 2000.

\bibitem{Wang1998CMP}
Shuzhou Wang.
\newblock Quantum symmetry groups of finite spaces.
\newblock {\em Commun. Math. Phys.}, 195:195--211, 1998.

\bibitem{wu2023testing}
Yihong Wu, Jiaming Xu, and Sophie~H Yu.
\newblock Testing correlation of unlabeled random graphs.
\newblock {\em The Annals of Applied Probability}, 33(4):2519--2558, 2023.

\bibitem{zhang2024quantum}
Xiao-Dong Zhang, Bin-Bin Cai, and Song Lin.
\newblock Quantum multigraph states and multihypergraph states.
\newblock {\em Physical Review A}, 109(6):062407, 2024.

\end{thebibliography}
\bibliographystyle{plain}

\appendix
\section{Validation of MC Sampling Strategy for Unlabeled Graphs}
Here we show a comparison of the MC and exact distributions of the internal energy per vertex $u$ and the specific heat per vertex $c$, fraction of vertices in the largest component $s_1$, and the susceptibility of $s_1$ for the unlabeled graphs for both the free and Ising Hamiltonians.
\begin{figure}[h!]
\centering
\begin{tabular}{cc}
    \includegraphics[width=0.42\textwidth]{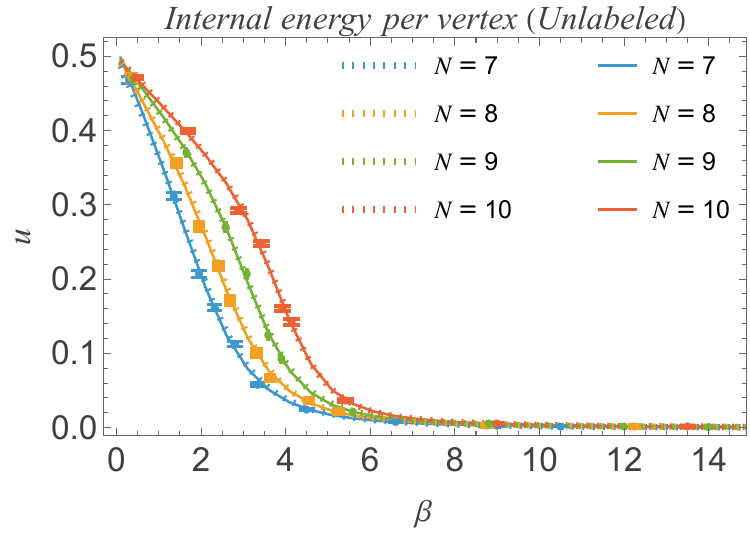}&\includegraphics[width=0.4\textwidth]{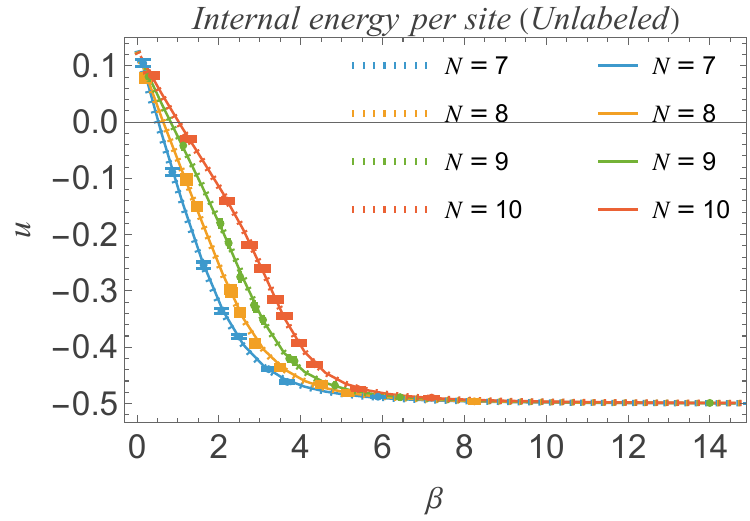}\\
    \includegraphics[width=0.4\textwidth]{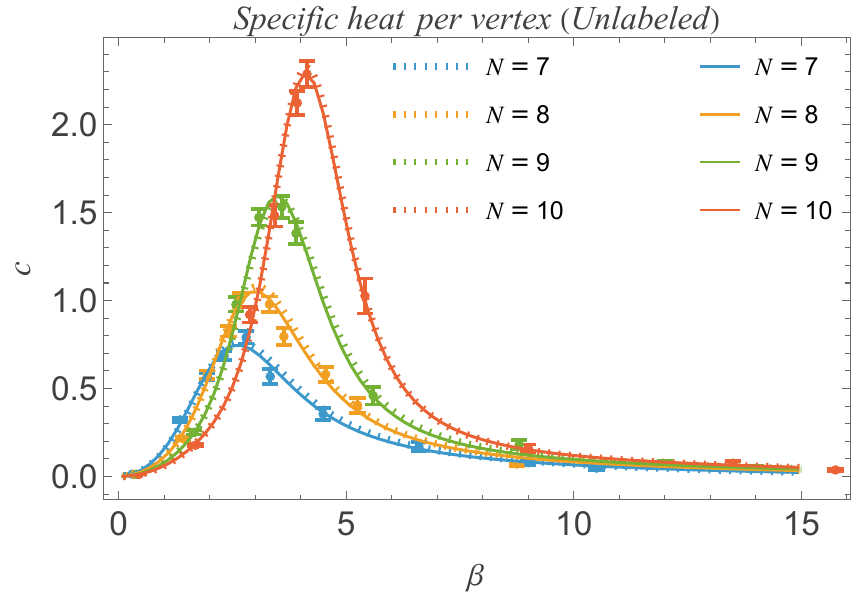}&\includegraphics[width=0.4\textwidth]{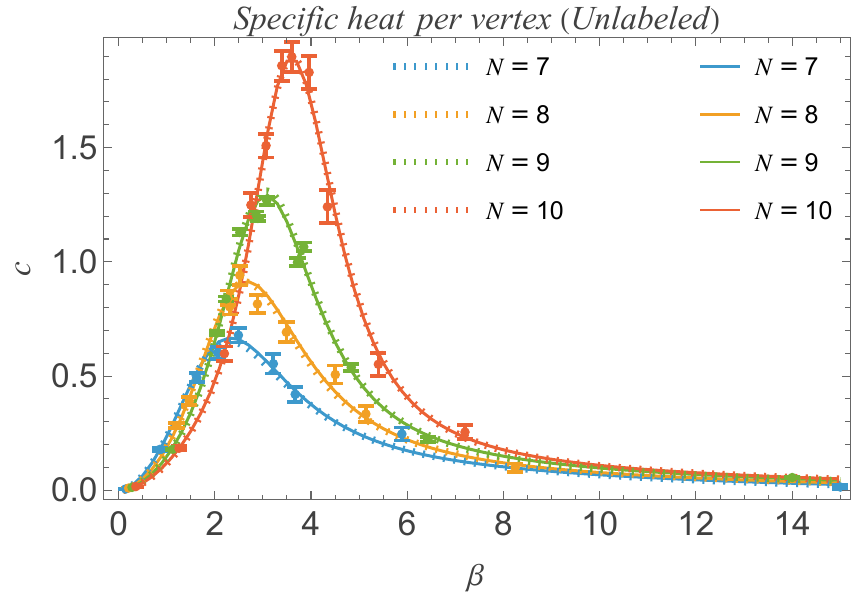}\\
    \includegraphics[width=0.4\textwidth]{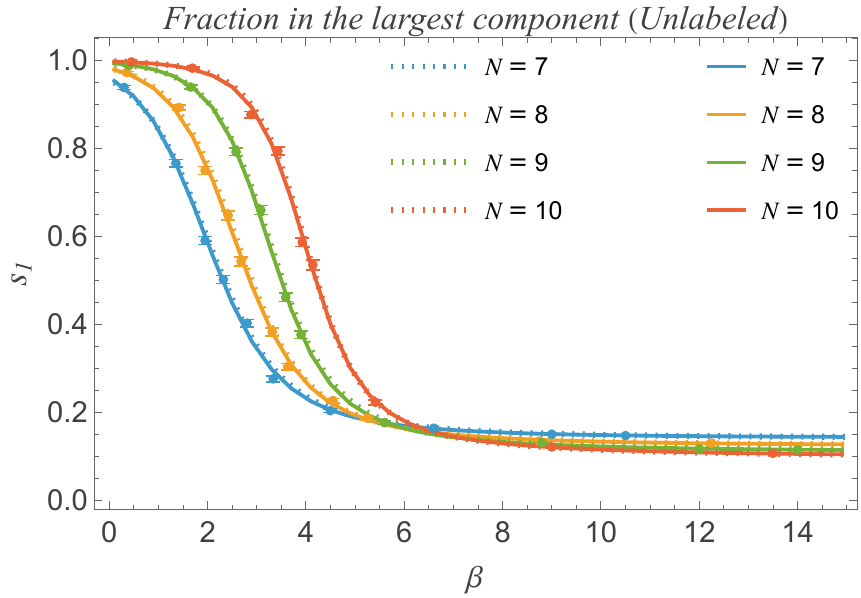}&\includegraphics[width=0.4\textwidth]{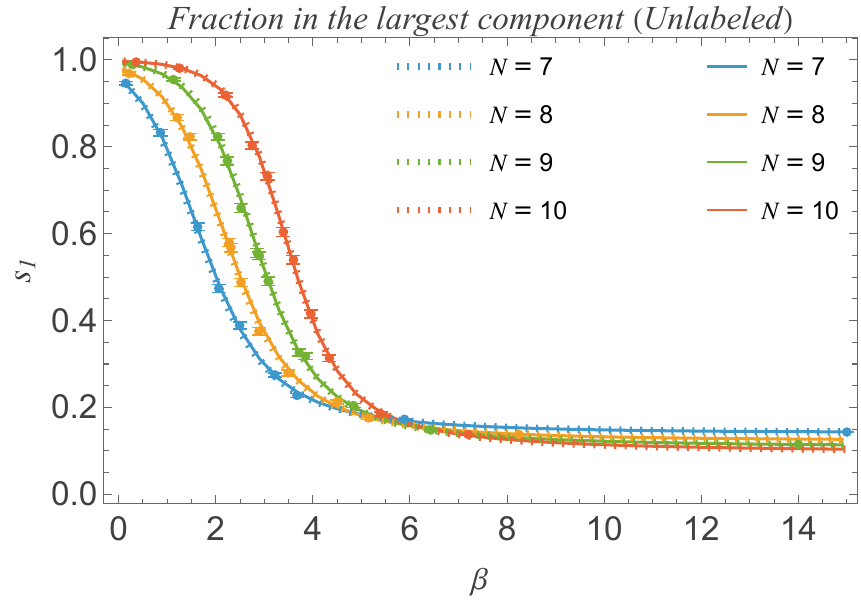}\\ 
    \includegraphics[width=0.4\textwidth]{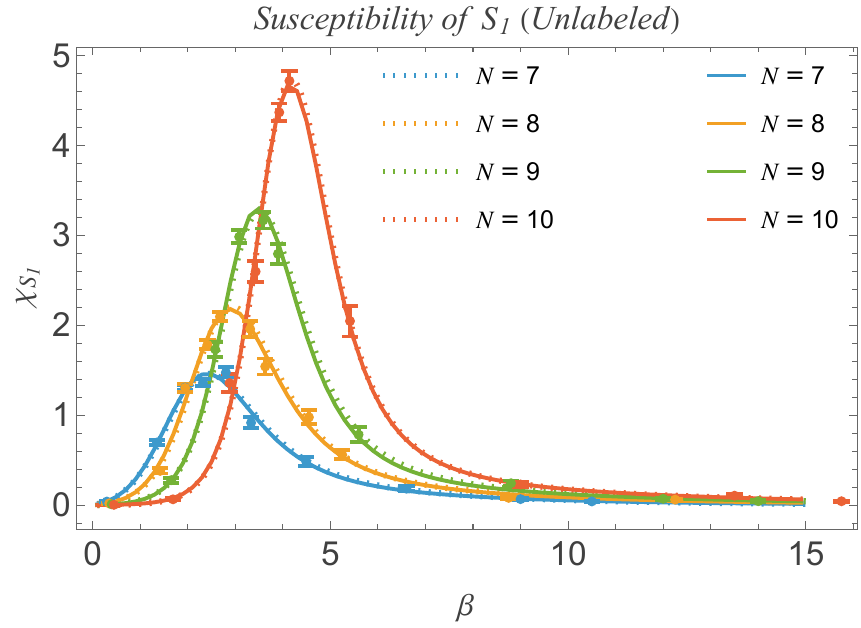}&\includegraphics[width=0.4\textwidth]{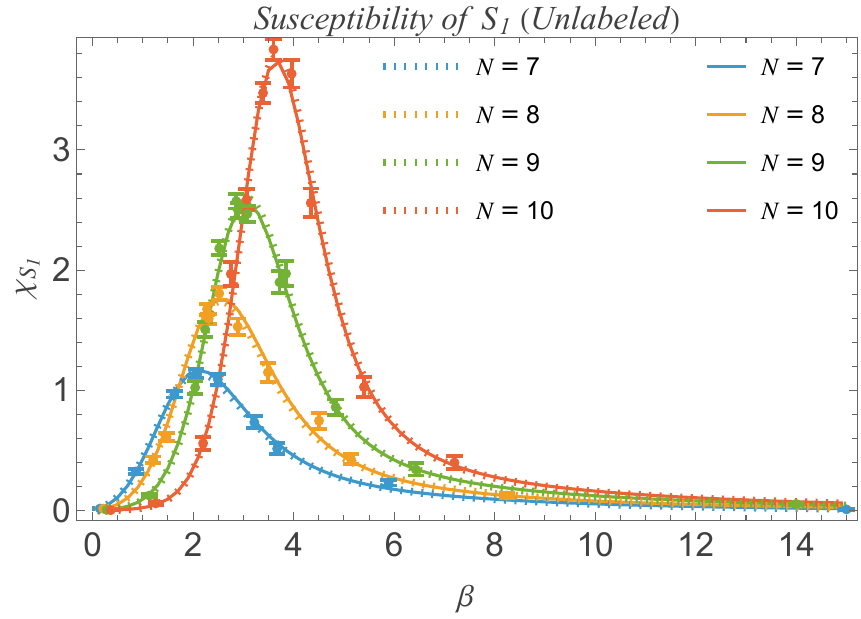} 

\end{tabular}
\caption{Graphs of $u, c, s_1,\chi_{s_1}$ as functions of inverse temperature $\beta$ for $N\in\{7,8,9,10\}$ for the free unlabeled graph system (left column) with $E_0=0,E_1=1$ and the Ising unlabeled graph system (right column) with $E_0=-0.5,E_1=1$. The dots with error bars are Monte Carlo simulations, solid lines are interpolations of these done using the multiple histograms extrapolation method. Dashed lines are exact computations. There is excellent agreement between the MC and exact plots.}
\label{fig:UnlabeledSamplingStrategyValidation}
\end{figure}
\end{document}